\documentclass[11pt,a4paper]{article}
\usepackage[utf8]{inputenc}
\usepackage[T1]{fontenc}
\usepackage{amsmath,amssymb,amsthm}
\usepackage{hyperref}
\usepackage{cleveref}
\usepackage{mathrsfs}
\usepackage{enumitem}
\usepackage{geometry}
\usepackage{comment}

\usepackage{tikz} \usepackage{booktabs} \usepackage{subcaption}

\numberwithin{equation}{section}

\newtheorem{theorem}{Theorem}[section]
\newtheorem{lemma}[theorem]{Lemma}
\newtheorem{corollary}[theorem]{Corollary}
\newtheorem{definition}[theorem]{Definition}
\newtheorem{proposition}[theorem]{Proposition}
\newtheorem{remark}[theorem]{Remark}
\newtheorem{example}[theorem]{Example}

\title{Self-extensional logics of formal inconsistency: Decidability and limits for paraconsistency}

\author{Marcelo E. Coniglio$^{1,2}$ and H\'ector Federico Mallea$^{1}$}
\date{
    $^\textup{\scriptsize 1}$\textit{\small Centre for Logic, Epistemology and the History of Science (CLE),}
    \\ 
    \textit{\small  University of Campinas (UNICAMP), Campinas, Brazil}
    \\
    $^\textup{\scriptsize 2}$\textit{\small Institute of Philosophy and the Humanities (IFCH),}
    \\ 
    \textit{\small  University of Campinas (UNICAMP), Campinas, Brazil}
   \\
   {\small e-mails: \texttt{coniglio@unicamp.br}; \ \  \texttt{hfmallea@gmail.com}}
   }

\begin{document}
\maketitle

\begin{abstract}
$\mathrm{RmbC}$ is a self-extensional paraconsistent logic in the family of Logics of Formal Inconsistency (LFIs). This system is obtained from $\mathrm{mbC}$ (the basic LFI) by adding the replacement property via two global inference rules. $\mathrm{RmbC}$ is characterized by a non-explosive negation $\neg$ and a consistency operator $\circ$, which recovers the principle of explosion in a controlled way. Together with its principal axiomatic extensions, $\mathrm{RmbC}$ admits a standard Lindenbaum--Tarski algebraization, with Boolean algebras with LFI operators (BALFIs) as its algebraic semantics.

In this paper, we study how far this self-extensional paraconsistent behavior can be extended axiomatically, starting from $\mathrm{RmbC}$. We classify pairs of very natural consistency axioms according to whether they preserve paraconsistency or force classical collapse; identify six minimal explosive combinations that collapse to a single algebraic core; and isolate a separate structural obstruction for the combination of excluded middle for $\neg$ with an involutive negation.

We also investigate, for the first time, the decidability of this family of self-extensional LFIs. As a first result, we prove the finite model property for $\mathrm{RmbC}$ with respect to BALFI semantics via an algebraic filtration, which yields decidability, and transfer this result to several paraconsistent axiomatic extensions of $\mathrm{RmbC}$.  Finally, we establish a \textsf{coNEXPTIME} upper bound for the validity problem of $\mathrm{RmbC}$ and a \textsf{coNP}-hardness lower bound, and prove \textsf{coNP}-completeness for the principal extensions containing one of the six minimal explosive pairs.
\end{abstract}

\section{Introduction}
\label{sec:introduction}

Reasoning in the presence of inconsistent information is a 
fundamental problem in logic. Paraconsistent systems were developed 
to avoid triviality in the presence of contradictions while still 
allowing controlled inference. Among them, the Logics of Formal 
Inconsistency (LFIs)~\cite{CM,CarnielliConiglioMarcos2007,CarnielliConiglio2016a} occupy a distinguished place because they internalize the notions of 
consistency and inconsistency at the object-language level, making it 
possible to distinguish harmless contradictions from explosive ones.

A major advance in the study of LFIs was the introduction of the 
self-extensional logic $\mathrm{RmbC}$~\cite{CarnielliConiglioFuenmayor2022} 
and its algebraic counterpart, the variety of Boolean algebras with 
LFI operators (BALFIs). Unlike most systems in the LFI hierarchy, 
$\mathrm{RmbC}$ satisfies the replacement property: equivalence is 
preserved by all connectives, which allows for a standard 
Lindenbaum--Tarski algebraization, something known to be impossible 
for da Costa's $C_1$ and every logic in his hierarchy weaker than it. 
The same methodology extends to a family of axiomatic enrichments of 
$\mathrm{RmbC}$.

This paper investigates how far this self-extensional, paraconsistent 
behavior can be pushed, along three complementary axes.

First, we ask which pairs of the main consistency axioms preserve paraconsistency and which force a collapse to classical logic. We identify six minimal explosive pairs algebraically, and use computer-assisted finite-model searches with Mace4 and Python to obtain finite paraconsistent witnesses for the remaining pairs, with the exception of $(\mathrm{ce},\mathrm{cf})$, for which finite paraconsistent models are impossible (although an infinite paraconsistent model exists~\cite{ConiglioFigalloTesta}). Thus, the paraconsistency behavior of all pairs of the fourteen principal axioms is completely determined.

Within this classification, we isolate a weakening of the consistency conditions associated with $C_1$ that remains paraconsistent under replacement --- in contrast to the self-extensional version of $C_1$ itself, which collapses to classical logic, as shown in~\cite{CarnielliConiglioFuenmayor2022}.

Second, this algebraic development raises a natural computational 
question, which remained unexplored up to this point, namely: whether $\mathrm{RmbC}$ and its principal paraconsistent extensions are decidable. This is a central issue, since decidability provides an effective procedure for theoremhood, clarifies proof search, and 
indicates whether these logics can be handled algorithmically in 
applications. We answer this question by proving the finite model 
property for $\mathrm{RmbC}$ with respect to BALFI semantics via an 
algebraic filtration, and by transferring this result to thirteen of 
the fourteen principal consistency extensions and to several families
of their combinations, with $(\mathrm{ce})$ as the sole principal
exception, using a common general construction with different support
and default configurations.

Third, we study the computational complexity of the validity problem 
for $\mathrm{RmbC}$. We establish a $\textsf{coNP}$-hardness lower bound
and a $\textsf{coNEXPTIME}$ upper bound, improving on the bound obtained
by a direct brute-force use of the finite-model construction.

The paper is organized as follows. Section~\ref{sec:balfi-prelim} recalls the background on LFIs, $\mathrm{RmbC}$, and 
BALFIs, and introduces the fourteen principal consistency axioms 
considered throughout. Section~\ref{sec:classification} carries out 
the classification described above, including the case of $C_1$. 
Section~\ref{sec:fmp} proves the finite model property and 
decidability for $\mathrm{RmbC}$. Section~\ref{sec:transfer} transfers 
this result to the principal extensions and discusses the boundary 
case of $(\mathrm{ce})$. Section~\ref{sec:complexity} establishes the 
complexity bounds. Section~\ref{sec:conclusion} summarizes the results 
and discusses directions for further work.

\section{Preliminaries}\label{sec:balfi-prelim}
\subsection{Notation and basic setting}

Throughout the paper, $\Theta$ denotes a propositional signature, and
$\mathrm{For}(\Theta)$ the absolutely free algebra of formulas over $\Theta$
generated by a fixed countably infinite set of propositional variables
$V=\{p_n : n\in \mathbb{N}\}$.
As usual, a logic will be identified with its Hilbert calculus presentation,
unless explicitly stated otherwise.

We shall consider the following signatures (for classical positive logic, Boolean algebras and LFIs, respectively):
\[
\Sigma^{+}=\{\wedge,\vee,\to\},
\qquad
\Sigma_{\mathrm{BA}}=\{\wedge,\vee,\to, 0, 1\},
\qquad
\Sigma=\{\wedge,\vee,\to,\neg,\circ\}.
\]

If $\varphi$ is a formula, we write $\mathrm{Var}(\varphi)$ for the set of
propositional variables occurring in $\varphi$, and $\mathrm{Sub}(\varphi)$
for the set of subformulas of $\varphi$. For a set $\Gamma$ of formulas,
$\mathrm{Var}(\Gamma)$ and $\mathrm{Sub}(\Gamma)$ are defined in the usual way.

Given a consequence relation $\vdash$, we write $\Gamma\vdash\varphi$ to
indicate that $\varphi$ is derivable from $\Gamma$, and simply $\vdash\varphi$
when $\varphi$ is a theorem. When the underlying logic is clear from the
context, subscripts will be omitted.

The algebraic semantics considered in this paper are based on Boolean algebras
expanded with additional unary operators interpreting the non-classical
connectives. By a harmless abuse of notation, an algebra and its underlying
set will occasionally be denoted by the same symbol.

When several algebras are considered simultaneously, subscripts on the
non-classical unary operators indicate the algebra in which they are
interpreted. Thus, $\neg_A$ and $\circ_A$ denote the corresponding
operators of $\mathbf{A}$, while $\neg_B$ and $\circ_B$ denote those of
$\mathbf{B}$. When the algebra under consideration is clear from the
context, the subscripts are omitted.

\subsection{Logics of Formal Inconsistency and the logic $\mathrm{RmbC}$}

Paraconsistency studies logical systems that can accommodate contradictory information without collapsing into triviality. Among such systems, the \emph{Logics of Formal Inconsistency} (LFIs), introduced in~\cite{CM} and further developed in~\cite{CarnielliConiglioMarcos2007, CarnielliConiglio2016a, avr:ari:zam:2018}, occupy a distinguished place, since they internalize at the object-language level the notions of consistency and inconsistency. This is achieved by means of a non-explosive negation together with a consistency operator that allows one to recover explosion in a controlled way.

\begin{definition}
Let $L=\langle \Theta,\vdash\rangle$ be a Tarskian, finitary, and structural logic over a propositional signature $\Theta$ containing a negation $\neg$, and let $\circ$ be a primitive or derived unary connective. Then $L$ is said to be a \emph{Logic of Formal Inconsistency} (LFI) with respect to $\neg$ and $\circ$ if the following conditions hold:
\begin{enumerate}
    \item $\varphi, \neg\varphi \not\vdash \psi$ for some $\varphi$ and $\psi$;
    \item $\circ\varphi, \varphi, \neg\varphi \vdash \psi$ for every $\varphi$ and $\psi$;
    \item there are two formulas $\alpha$ and $\beta$ such that
    \begin{enumerate}
        \item $\circ\alpha, \alpha 
        \not\vdash \beta$;
        \item $\circ\alpha, \neg\alpha \not\vdash \beta$.
    \end{enumerate}
\end{enumerate}
\end{definition}

Condition (1) above signals the non-validity of the classical principle of explosion. Condition (2), also called principle of gentle explosion, characterizes LFIs in particular. Condition (3) is required in order to satisfy condition (2) in a non-trivial way. 

The hierarchy of LFIs studied in \cite{CarnielliConiglioMarcos2007} begins with a logic called mbC, which extends positive classical logic CPL+ by adding a negation $\neg$ and a primitive consistency operator $\circ$ that meets the minimal requirements to define an LFI.

\begin{definition} \label{def:mbC}
    The classical positive logic CPL+ is defined by the following Hilbert calculus:
\begin{itemize}
    \item \textbf{Axiom schemas:}
    \begin{align*}
        &\alpha \rightarrow (\beta\rightarrow\alpha) \quad &&(Ax1) \\
        &(\alpha \rightarrow (\beta \rightarrow \gamma)) \rightarrow ((\alpha \rightarrow \beta) \rightarrow (\alpha \rightarrow \gamma)) \quad &&(Ax2) \\
        &\alpha \rightarrow (\beta \rightarrow (\alpha \land \beta)) \quad &&(Ax3) \\
        &(\alpha \land \beta) \rightarrow \alpha \quad &&(Ax4) \\
        &(\alpha \land \beta) \rightarrow \beta \quad &&(Ax5) \\
        &\alpha \rightarrow (\alpha \lor \beta) \quad &&(Ax6) \\
        &\beta \rightarrow (\alpha \lor \beta) \quad &&(Ax7) \\
        &(\alpha \rightarrow \gamma) \rightarrow ((\beta \rightarrow \gamma) \rightarrow ((\alpha \lor \beta) \rightarrow \gamma))  \quad &&(Ax8) \\
        &(\alpha \rightarrow \beta) \lor \alpha \quad &&(Ax9)
    \end{align*}
    \item \textbf{Inference rule:}
    \begin{align*}
        &\frac{\alpha \quad \alpha \rightarrow \beta}{\beta} \quad &&(\text{MP})
    \end{align*}
\end{itemize}
\end{definition}
\begin{definition}
    The logic mbC is obtained from CPL+ by adding the following axiom schemas:
\begin{align*}
    &\alpha \lor \neg \alpha &&\quad (Ax10) \\
    &{\circ} \alpha \rightarrow (\alpha \rightarrow (\neg \alpha \rightarrow \beta)) &&\quad (\text{bc1})
\end{align*}
\end{definition}
The logic $\mathrm{mbC}$ is the minimal LFI extending $\mathrm{CPL}^+$~\cite{CarnielliConiglioMarcos2007}. It is well known, however, that $\mathrm{mbC}$ does not satisfy the replacement property in full generality. Replacement is available for every formula over $\Sigma^+$, and this observation motivates the passage to the self-extensional extension $\mathrm{RmbC}$, obtained by requiring replacement also for the non-classical connectives $\neg$ and $\circ$.

\begin{definition}[\cite{CarnielliConiglioFuenmayor2022}]
    The logic $\mathrm{RmbC}$ is obtained from $\mathrm{mbC}$ by adding the following inference rules:
$$\displaystyle \frac{\alpha \leftrightarrow \psi}{\neg\alpha \leftrightarrow \neg\psi} \quad (R_{\neg}) \hspace{2cm} \displaystyle \frac{\alpha \leftrightarrow \psi}{{\circ}\alpha \leftrightarrow {\circ}\psi} \quad (R_{\circ})$$
   
\end{definition}

In each LFI $L$ extended with the two rules above, derivability from premises is understood in the usual local sense:  we write $\Gamma\vdash_L\varphi$ if either $\vdash_L\varphi$ (i.e., $\varphi$ is a theorem of $L$ in the standard sense), or there exists a finite non-empty subset $\{\gamma_1,\ldots,\gamma_n\}\subseteq\Gamma$ such that
\[
\vdash_L(\gamma_1\wedge\cdots\wedge\gamma_n)\to\varphi.
\]

Thus, the rules $(R_{\neg})$ and $(R_{\circ})$ are global in the sense that their premises must be theorems, not merely assumptions. This is the key feature that determines the degree-preserving character of the algebraic semantics introduced in the next subsection. As shown in~\cite{CarnielliConiglioFuenmayor2022}, $\mathrm{RmbC}$ satisfies the replacement property and remains paraconsistent. The self-extensional behavior of $\mathrm{RmbC}$ is what makes it, together with its principal axiomatic extensions which are still paraconsistent, amenable to a standard Lindenbaum--Tarski algebraization. The corresponding algebraic semantics is provided by Boolean algebras with LFI operators, introduced in the next subsection.
\subsection{The algebraic semantics for \texorpdfstring{$\mathrm{RmbC}$}{RmbC}}
\label{sec:balfi}

The algebraic semantics for $\mathrm{RmbC}$ is given by a class of Boolean 
algebras expanded with two additional unary operators interpreting the 
non-classical connectives $\neg$ and $\circ$. These operators are called 
\emph{LFI operators}, and the resulting structures are termed \emph{Boolean 
algebras with LFI operators}, or \emph{BALFIs} for short.

\begin{definition}
A \emph{BALFI} is an algebra
\[
\mathbf{B}=\langle B,\wedge,\vee,\to,\neg,\circ,0,1\rangle
\]
such that its reduct
\[
\langle B,\wedge,\vee,\to,0,1\rangle
\]
is a Boolean algebra and, for every $a\in B$, the following equations hold:
\[
a\vee \neg a = 1,
\qquad
a\wedge \neg a \wedge \circ a = 0.
\]
The class of all BALFIs will be denoted by $\mathsf{BI}$.
\end{definition}

The first equation is a form of excluded middle for the non-classical negation; 
the second captures the controlled explosion mechanism: the conjunction of any 
element with both its negation and its consistency value is always zero. 
Together, they encode the minimal requirements of an LFI at the equational level.

Since every BALFI carries two distinct negation-like operations --- the 
Boolean complement ${\sim}x:=x \to 0$ and the 
non-classical operator $\neg$ --- some care is needed to avoid conflating 
the two. Indeed, they coincide only on a ``well-behaved'' part of the algebra (adopting da Costa's terminology), and precisely distinguishing them is a recurring concern throughout the 
paper. The following lemma records their most basic relationship.
\begin{lemma}\label{lem:neg-complement-bound}
In every BALFI $\mathbf{B}$, ${\sim}x\leq\neg_B x$ for every $x\in B$.
\end{lemma}

\begin{proof}
By the BALFI axiom, $x\vee\neg_B x=1$. In any Boolean algebra, 
$x\vee y=1$ if and only if ${\sim}x\leq y$. Taking $y=\neg_B x$ gives 
${\sim}x\leq\neg_B x$.
\end{proof}

A \emph{valuation} over a BALFI $\mathbf{B}$ is a homomorphism $v:\mathrm{For}(\Sigma)\to B$. 
The values of compound formulas are thus computed compositionally, starting from an assignment $v_0:V \to B$. A formula 
$\varphi$ is \emph{valid in} $\mathsf{BI}$, written $\vDash_{\mathsf{BI}}\varphi$, 
if $v(\varphi)=1$ for every BALFI $\mathbf{B}$ and every valuation $v$ on 
$\mathbf{B}$.

The consequence relation associated with this semantics is 
\emph{degree-preserving} (or local): for a set of formulas 
$\Gamma\cup\{\varphi\}$, we write $\Gamma\vDash_{\mathsf{BI}}\varphi$ if 
either $\vDash_{\mathsf{BI}}\varphi$, or there exists a finite non-empty 
subset $\{\gamma_1,\ldots,\gamma_n\}\subseteq\Gamma$ such that, for every 
BALFI $\mathbf{B}$ and every valuation $v$ on $\mathbf{B}$,
\[
\bigwedge_{i=1}^{n}v(\gamma_i)\leq v(\varphi).
\]
Equivalently,
\[
\Gamma\vDash_{\mathsf{BI}}\varphi
\quad\text{if and only if}\quad
\vDash_{\mathsf{BI}}(\gamma_1\wedge\cdots\wedge\gamma_n)\to\varphi
\]
for some such finite subset. This local reading matches the proof-theoretic 
treatment of derivations from premises in $\mathrm{RmbC}$, where the global 
rules $(R_{\neg})$ and $(R_{\circ})$ restrict their premises to theorems.

The following soundness and  completeness theorem, established 
in~\cite{CarnielliConiglioFuenmayor2022}, is the algebraic foundation of the 
present paper.

\begin{theorem}[\cite{CarnielliConiglioFuenmayor2022}]
\label{thm:completeness-RmbC}
For every $\Gamma\cup\{\varphi\}\subseteq\mathrm{For}(\Sigma)$,
\[
\Gamma\vdash_{\mathrm{RmbC}}\varphi
\quad\text{if and only if}\quad
\Gamma\vDash_{\mathsf{BI}}\varphi.
\]
\end{theorem}

As a consequence, $\mathrm{RmbC}$ admits a standard algebraic semantics in 
the sense of Lindenbaum--Tarski: the Lindenbaum--Tarski algebra of 
$\mathrm{RmbC}$ is a BALFI, and the canonical valuation sends each formula 
to its equivalence class under provable equivalence. The self-extensionality 
of $\mathrm{RmbC}$, guaranteed by the rules $(R_{\neg})$ and 
$(R_{\circ})$, is precisely what makes this construction possible, since it 
ensures that provable equivalence is a congruence with respect to all 
connectives.
\subsection{Principal axiomatic extensions}
\label{sec:extensions}

The framework introduced above admits a family of axiomatic enrichments.
These extensions are the natural objects to study once the replacement
property has been secured, since the added axioms interact smoothly with
the self-extensional behavior of the logic. We shall consider the following axiom schemata (in $(ca_{\#})$, $\#\in\{\wedge,\vee,\to\}$):\\[1mm]

$\begin{array}{llll}
        (ciw) &\quad  \circ\alpha \vee (\alpha \wedge \neg\alpha) &  \hspace{2cm} (cl) &\quad  \neg(\alpha \wedge \neg\alpha) \to \circ\alpha \\
        (d1) &\quad  \circ\alpha\vee\alpha  &  \hspace{2cm} (ce)&\quad  \alpha \to \neg\neg\alpha \\
        (d2) &\quad  \circ\alpha\vee\neg\alpha  &  \hspace{2cm} (cf) &\quad  \neg\neg\alpha \to \alpha \\
        (ci) &\quad \neg{\circ}\alpha \to (\alpha \wedge \neg\alpha)  &  \hspace{2cm}  (a_{\circ}) &\quad  \circ{\circ}\alpha \\
        (i1)&\quad  \neg{\circ}\alpha\to\alpha  &  \hspace{2cm} (a_{\neg}) &\quad  \circ\alpha \to \circ\neg\alpha \\
        (i2) &\quad  \neg{\circ}\alpha\to\neg\alpha  &  \hspace{2cm}  (ca_{\#}) &\quad (\circ\alpha \wedge \circ\beta)\to \circ(\alpha\#\beta)
\end{array}$

\begin{remark}
The axiom schemata considered in this section were presented in the study of principal extensions of $\mathrm{mbC}$ by Carnielli and Coniglio in \cite{CarnielliConiglio2016a}. The four schemata corresponding to the conditions $(d1)$, $(d2)$, $(i1)$, and $(i2)$ go back to Avron's modular treatment of LFIs in \cite{Avron}.
\end{remark}

Let $Ax$ be any nonempty set of axioms chosen from the collection above,
and let $\mathrm{RmbC}(Ax)$ denote the logic obtained by adding $Ax$ to
$\mathrm{RmbC}$. Likewise, let $\mathsf{BI}(Ax)$ be the class of BALFIs
validating every axiom in $Ax$. For brevity, we write $\mathsf{BI}(A_1,\ldots,A_k)$ for $\mathsf{BI}(\{A_1,\ldots,A_k\})$, and similarly for
$\mathrm{RmbC}(A_1,\ldots,A_k)$.

Since each axiom corresponds to an equation
on the algebraic side, $\mathsf{BI}(Ax)$ is a variety. The following
uniform soundness and completeness theorem, established
in~\cite{CarnielliConiglioFuenmayor2022},\footnote{Rigorously speaking, axioms (d1),(d2), (i1), (i2), (a$_\circ$) and (a$_\neg$) were not considered in~\cite{CarnielliConiglioFuenmayor2022}. However, they can be easily analyzed by adapting the techniques used for the other axioms.} provides the algebraic semantics
needed for the finite model arguments developed in
Sections~\ref{sec:fmp} and~\ref{sec:transfer}.

\begin{theorem}[\cite{CarnielliConiglioFuenmayor2022}]
\label{thm:completeness-uniform}
For every nonempty $Ax$ as above and every
$\Gamma\cup\{\varphi\}\subseteq\mathrm{For}(\Sigma)$,
\[
\Gamma\vdash_{\mathrm{RmbC}(Ax)}\varphi
\quad\text{if and only if}\quad
\Gamma\vDash_{\mathsf{BI}(Ax)}\varphi.
\]
\end{theorem}

The algebraic counterparts of the axioms are the following equations.
For every BALFI
$\mathbf{B}=\langle B,\wedge,\vee,\to,\neg,\circ,0,1\rangle$
and every $a,b\in B$, where ${\sim}$ denotes the Boolean complement:\\[1mm]

$\begin{array}{ll}
        \mathbf{B}\in\mathsf{BI}(\mathrm{ciw}) \ \iff \ \circ a = {\sim}(a\wedge\neg a), &  \hspace{2cm} \mathbf{B}\in\mathsf{BI}(\mathrm{cl}) \ \iff \ \circ a = \neg(a\wedge\neg a),\\
        \mathbf{B}\in\mathsf{BI}(d1) \ \iff \ \circ a\vee a = 1,  &  \hspace{2cm} \mathbf{B}\in\mathsf{BI}(\mathrm{ce})  \iff \ a\wedge\neg\neg a =  a,\\
        \mathbf{B}\in\mathsf{BI}(d2) \ \iff \ \circ a\vee\neg a = 1,  &  \hspace{2cm} \mathbf{B}\in\mathsf{BI}(\mathrm{cf}) \ \iff \ a\wedge\neg\neg a = \neg\neg a,\\
        \mathbf{B}\in\mathsf{BI}(\mathrm{ci}) \ \iff \ \neg(\circ a) = a\wedge\neg a,  &  \hspace{2cm}  \mathbf{B}\in\mathsf{BI}(\mathrm{a}_\circ) \ \iff \ \circ{\circ} a = 1,\\
        \mathbf{B}\in\mathsf{BI}(i1) \ \iff \ \neg{\circ} a\wedge a = \neg{\circ} a,  &  \hspace{2cm} \mathbf{B}\in\mathsf{BI}(\mathrm{a}_\neg) \ \iff \ \circ a\wedge\circ\neg a = \circ a,\\
        \mathbf{B}\in\mathsf{BI}(i2) \ \iff \ \neg{\circ} a\wedge\neg a = \neg{\circ} a,  &  \hspace{2cm}  \mathbf{B}\in\mathsf{BI}(\mathrm{ca}_\#)  \iff  (\circ a\wedge\circ b)\wedge\circ(a\# b)\\
        &  \hspace{6.5cm} = \circ a\wedge\circ b.
\end{array}$

\ 

The inclusion relations, incomparabilities, and explosive combinations
among these varieties are analyzed in
Section~\ref{sec:classification}.
\subsection{The finite model property and decidability}
\label{sec:fmp-prelim}

A logic $L$ with semantics given by a class of algebras $\mathsf{K}$ is said
to have the \emph{finite model property} (FMP) if, whenever
$\nvdash_L\varphi$, there exist a finite algebra
$\mathbf{A}\in\mathsf{K}$ and a valuation $v$ on $\mathbf{A}$ such that $v(\varphi)\neq 1$.

The connection between the FMP and decidability is well-known. A logic $L$ is
said to be \emph{decidable} if the set of its theorems is recursive. The
following result, due to Harrop~\cite{Harrop1958}, makes the FMP
algorithmically relevant.

\begin{theorem}[Harrop~\cite{Harrop1958}]
\label{thm:harrop}
Let $L$ be a finitely axiomatizable logic having the finite model property
with respect to a recursively enumerable class of finite algebras with
decidable validity. Then $L$ is decidable.
\end{theorem}

In practice, the finite axiomatization provides an effective enumeration of
proofs, while the FMP guarantees that every non-theorem admits a finite
countermodel. The combination of these two procedures yields a decision
algorithm for theoremhood. This is the approach adopted in
Sections~\ref{sec:fmp} and~\ref{sec:transfer} to establish the decidability
of $\mathrm{RmbC}$ and its principal extensions.
\section{Subvarieties of BALFIs and the limits of paraconsistency}
\label{sec:classification}

\subsection{The variety landscape: subvarieties and inclusion relations}
\label{sec:varieties}

The axiom schemata introduced in Section~\ref{sec:extensions} generate a 
family of subvarieties of $\mathsf{BI}$. Since each axiom corresponds to an 
equation on the algebraic side, the class $\mathsf{BI}(Ax)$ is a variety for 
every nonempty $Ax$. The inclusion and incomparability relations among the extensions of mbC by adding these axioms were already studied in the literature~\cite{CarnielliConiglio2016a, Avron}, and the same 
relations persist in the present BALFI setting because the algebraic 
counterparts are unchanged. In this subsection we record the relations that 
will be used in Section~\ref{sec:paraconsistency} to analyze the explosive 
combinations of axioms.

Throughout this subsection, all BALFIs are defined over the four-element 
Boolean algebra $\wp(\{w_1,w_2\})$, with elements $0=\emptyset$, 
$a=\{w_1\}$, $b=\{w_2\}$, and $1=\{w_1,w_2\}$.

\subsubsection*{Relations among \texorpdfstring{$(d1)$, $(d2)$, $(i1)$, 
$(i2)$, $(\mathrm{ciw})$, and $(\mathrm{ci})$}{d1, d2, i1, i2, ciw, ci}}

The following relations hold:
\begin{enumerate}
\item $\mathsf{BI}(i1)\subseteq\mathsf{BI}(d1)$ and 
      $\mathsf{BI}(i2)\subseteq\mathsf{BI}(d2)$
      \cite{CarnielliConiglio2016a, Avron};
\item $\mathsf{BI}(d1)\cap\mathsf{BI}(d2)=\mathsf{BI}(\mathrm{ciw})$
      \cite{CarnielliConiglio2016a, Avron};
\item $\mathsf{BI}(i1)\cap\mathsf{BI}(i2)=\mathsf{BI}(\mathrm{ci})$
      \cite{Avron};
\item $\mathsf{BI}(\mathrm{ci})\subseteq\mathsf{BI}(\mathrm{ciw})$
      \cite{CarnielliConiglio2016a};
\item $\mathsf{BI}(d1)$ and $\mathsf{BI}(d2)$ are incomparable, as 
      witnessed by Examples~\ref{ex:D1} and~\ref{ex:D2} below;
\item $\mathsf{BI}(i1)$ and $\mathsf{BI}(i2)$ are incomparable, as witnessed by Examples~\ref{ex:E1} and~\ref{ex:E2} below;
\item $\mathsf{BI}(i1)$ and $\mathsf{BI}(\mathrm{ciw})$ are 
      incomparable, as witnessed by Example~\ref{ex:E1} below and 
      by Example~\ref{ex:B1};
\item $\mathsf{BI}(i2)$ and $\mathsf{BI}(\mathrm{ciw})$ are 
      incomparable, as witnessed by Example~\ref{ex:E2} below and 
      by Example~\ref{ex:B1}.
\end{enumerate}

Items~(7) and~(8) are new: although 
$\mathsf{BI}(\mathrm{ci})=\mathsf{BI}(i1)\cap\mathsf{BI}(i2)$ is 
contained in $\mathsf{BI}(\mathrm{ciw})$ by item~(4), neither $i1$ 
nor $i2$ individually implies $(\mathrm{ciw})$, and $(\mathrm{ciw})$ 
implies neither $i1$ nor $i2$.

\begin{example}\label{ex:D1}
Let $\mathbf{D}_1$ be the BALFI defined by:
\[
\begin{array}{|c|c|c|}
\hline
x & \neg x & \circ x \\
\hline
0 & 1 & 1 \\
a & 1 & b \\
b & a & a \\
1 & 0 & 1 \\
\hline
\end{array}
\]
One verifies that $\mathbf{D}_1\in\mathsf{BI}(d1)$: indeed, 
$\circ x\vee x=1$ for all $x$. However, 
$\circ b\vee\neg b=a\vee a=a\neq 1$, so 
$\mathbf{D}_1\notin\mathsf{BI}(d2)$.
\end{example}

\begin{example}\label{ex:D2}
Let $\mathbf{D}_2$ be the BALFI defined by:
\[
\begin{array}{|c|c|c|}
\hline
x & \neg x & \circ x \\
\hline
0 & 1 & 1 \\
a & b & a \\
b & 1 & 0 \\
1 & 0 & 1 \\
\hline
\end{array}
\]
One verifies that $\mathbf{D}_2\in\mathsf{BI}(d2)$: indeed, 
$\circ x\vee\neg x=1$ for all $x$. However, 
$\circ a\vee a=a\vee a=a\neq 1$, so 
$\mathbf{D}_2\notin\mathsf{BI}(d1)$.
\end{example}

\begin{example}\label{ex:E1}
Let $\mathbf{E}_1$ be the BALFI defined by:
\[
\begin{array}{|c|c|c|}
\hline
x & \neg x & \circ x \\
\hline
0 & 1 & 1 \\
a & b & b \\
b & a & a \\
1 & 0 & 0 \\
\hline
\end{array}
\]
One verifies that $\mathbf{E}_1\in\mathsf{BI}(i1)$: indeed, 
$\neg(\circ x)\leq x$ holds for all $x$ (moreover,  
$\neg(\circ x)=x$ for every $x$). However, 
$a\wedge\neg a=a \wedge b=0$, while $\circ a=b\neq 1$, so 
$\mathbf{E}_1\notin\mathsf{BI}(\mathrm{ciw})$.
Moreover, $\mathbf{E}_1\notin\mathsf{BI}(i2)$, since
$
\neg(\circ a)\wedge\neg a
= a\wedge b
=0
\neq a
=\neg(\circ a).$
\end{example}

\begin{example}\label{ex:E2}
Let $\mathbf{E}_2$ be the BALFI defined by:
\[
\begin{array}{|c|c|c|}
\hline
x & \neg x & \circ x \\
\hline
0 & 1 & 1 \\
a & b & a \\
b & a & b \\
1 & 0 & 1 \\
\hline
\end{array}
\]
It is immediate to see that $\mathbf{E}_2\in\mathsf{BI}(i2)$: indeed, 
$\neg(\circ x)\leq\neg x$ holds for all $x$. However,  $a\wedge\neg a=0$ while 
$\circ a=a\neq 1$, so $\mathbf{E}_2\notin\mathsf{BI}(\mathrm{ciw})$.
Moreover, $\mathbf{E}_2\notin\mathsf{BI}(i1)$, since $\neg(\circ a)\wedge a
= b\wedge a
=0
\neq b
=\neg(\circ a)$.
\end{example}

The converse direction is witnessed by $\mathbf{B}_1$ 
(Example~\ref{ex:B1} below): one checks directly from its table that 
$\mathbf{B}_1\in\mathsf{BI}(\mathrm{ciw})$, while 
$\neg(\circ_1 0)=a\not\leq 0$ refutes $(i1)$ and 
$\neg(\circ_1 a)=a\not\leq\neg_1 a=b$ refutes $(i2)$, so 
$\mathbf{B}_1\notin\mathsf{BI}(i1)$ and 
$\mathbf{B}_1\notin\mathsf{BI}(i2)$.

These relations are summarized in Figure~\ref{fig:hasse1}.

\begin{figure}[ht]
\centering
\begin{tikzpicture}[
  every node/.style={font=\small, inner sep=4pt},
  edge/.style={thick, shorten >=2pt, shorten <=2pt}
]
\node (ci)  at ( 0.0, 0.0) {$\mathsf{BI}(\mathrm{ci})$};
\node (i1)  at (-3.0, 2.5) {$\mathsf{BI}(i1)$};
\node (ciw) at ( 0.0, 2.5) {$\mathsf{BI}(\mathrm{ciw})$};
\node (i2)  at ( 3.0, 2.5) {$\mathsf{BI}(i2)$};
\node (d1)  at (-3, 5.0) {$\mathsf{BI}(d1)$};
\node (d2)  at ( 3, 5.0) {$\mathsf{BI}(d2)$};
\node (BI)  at ( 0.0, 7.0) {$\mathsf{BI}$};

\draw[edge] (ci)  -- (i1);
\draw[edge] (ci)  -- (ciw);
\draw[edge] (ci)  -- (i2);
\draw[edge] (i1)  -- (d1);
\draw[edge] (ciw) -- (d1);
\draw[edge] (ciw) -- (d2);
\draw[edge] (i2)  -- (d2);
\draw[edge] (d1)  -- (BI);
\draw[edge] (d2)  -- (BI);
\end{tikzpicture}
\caption{Inclusion diagram for the subvarieties of $\mathsf{BI}$ 
generated by $(d1)$, $(d2)$, $(i1)$, $(i2)$, $(\mathrm{ciw})$, and 
$(\mathrm{ci})$. The three varieties $\mathsf{BI}(i1)$, 
$\mathsf{BI}(\mathrm{ciw})$, and $\mathsf{BI}(i2)$ each contain 
$\mathsf{BI}(\mathrm{ci})$ but are pairwise incomparable (Examples~\ref{ex:E1}, \ref{ex:E2}, and~\ref{ex:B1}). Edges 
indicate the inclusion relations established in the text.}
\label{fig:hasse1}
\end{figure}
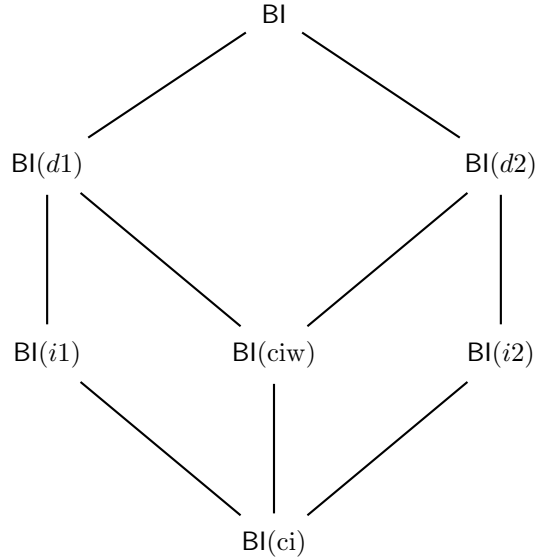

\subsubsection*{Relations among \texorpdfstring{$(\mathrm{ciw})$, 
$(\mathrm{ci})$, $(\mathrm{cl})$, $(\mathrm{cf})$, and 
$(\mathrm{ce})$}{ciw, ci, cl, cf, ce}}

The following relations hold:
\begin{enumerate}
\item $\mathsf{BI}(\mathrm{ci})\subseteq\mathsf{BI}(\mathrm{ciw})$ and 
      $\mathsf{BI}(\mathrm{cl})\subseteq\mathsf{BI}(\mathrm{ciw})$
      \cite{CarnielliConiglio2016a};
\item $\mathsf{BI}(\mathrm{ci})$ and $\mathsf{BI}(\mathrm{cl})$ are 
      incomparable, as witnessed by Example~\ref{ex:B2B3};
\item $\mathsf{BI}(\mathrm{cf})$ and $\mathsf{BI}(\mathrm{ciw})$ are 
      incomparable, as witnessed by Examples~\ref{ex:B1} 
      and~\ref{ex:B4};
\item $\mathsf{BI}(\mathrm{ce})$ is incomparable with each of 
      $\mathsf{BI}(\mathrm{ciw})$, $\mathsf{BI}(\mathrm{ci})$, 
      $\mathsf{BI}(\mathrm{cl})$, and $\mathsf{BI}(\mathrm{cf})$, as 
      witnessed by Example~\ref{ex:B5};
\item $\mathsf{BI}(\{\mathrm{ci},\mathrm{cf}\})=
      \mathsf{BI}(\{\mathrm{cl},\mathrm{cf}\})=
      \mathsf{BI}(\{\mathrm{ci},\mathrm{cl},\mathrm{cf}\})$
      \cite{CarnielliConiglioFuenmayor2022}.
\end{enumerate}

\begin{example}\label{ex:B1}
Let $\mathbf{B}_1$ be the BALFI defined by:
\[
\begin{array}{|c|c|c|}
\hline
x & \neg x & \circ x \\
\hline
0 & 1 & 1 \\
a & b & 1 \\
b & 1 & a \\
1 & a & b \\
\hline
\end{array}
\]
One verifies $\mathbf{B}_1\in\mathsf{BI}(\mathrm{ciw})$. Moreover, 
$a\wedge\neg a=0\neq a=\neg(\circ a)$, so 
$\mathbf{B}_1\notin\mathsf{BI}(\mathrm{ci})$; 
$\neg(b\wedge\neg b)=1\neq a=\circ b$, so 
$\mathbf{B}_1\notin\mathsf{BI}(\mathrm{cl})$; and 
$0\wedge\neg\neg 0=0\neq a=\neg\neg 0$, so 
$\mathbf{B}_1\notin\mathsf{BI}(\mathrm{cf})$.
\end{example}

\begin{example}\label{ex:B2B3}
Let $\mathbf{B}_2$ and $\mathbf{B}_3$ be the BALFIs defined by:
\[
\begin{array}{|c|c|c|}
\hline
x & \neg_{\mathbf{B}_2}  x & \circ_{\mathbf{B}_2} x \\
\hline
0 & 1 & 1 \\
a & 1 & b \\
b & a & 1 \\
1 & 0 & 1 \\
\hline
\end{array}
\qquad
\begin{array}{|c|c|c|}
\hline
x & \neg_{\mathbf{B}_3} x & \circ_{\mathbf{B}_3} x \\
\hline
0 & 1 & 1 \\
a & b & 1 \\
b & a & 1 \\
1 & a & b \\
\hline
\end{array}
\]
Then $\mathbf{B}_2\in\mathsf{BI}(\mathrm{ci})$ but 
$\circ_{\mathbf{B}_2} a=b\neq 1=\neg_{\mathbf{B}_2}(a\wedge\neg_{\mathbf{B}_2} a)$, so 
$\mathbf{B}_2\notin\mathsf{BI}(\mathrm{cl})$; while 
$\mathbf{B}_3\in\mathsf{BI}(\mathrm{cl})$ but 
$a\wedge\neg_{\mathbf{B}_3} a=0\neq a=\neg_{\mathbf{B}_3}(\circ_{\mathbf{B}_3} a)$, so 
$\mathbf{B}_3\notin\mathsf{BI}(\mathrm{ci})$.
\end{example}

\begin{example}\label{ex:B4}
Let $\mathbf{B}_4$ be the BALFI defined by:
\[
\begin{array}{|c|c|c|}
\hline
x & \neg x & \circ x \\
\hline
0 & 1 & 1 \\
a & 1 & 0 \\
b & 1 & 0 \\
1 & 0 & 1 \\
\hline
\end{array}
\]
Then $\mathbf{B}_4\in\mathsf{BI}(\mathrm{cf})$ but 
${\sim}(a\wedge\neg a)=b\neq 0=\circ a$, so 
$\mathbf{B}_4\notin\mathsf{BI}(\mathrm{ciw})$.
\end{example}

\begin{example}\label{ex:B5}
Let $\mathbf{B}_5$ be the BALFI defined by:
\[
\begin{array}{|c|c|c|}
\hline
x & \neg x & \circ x \\
\hline
0 & 1 & 1 \\
a & 1 & 0 \\
b & 1 & 0 \\
1 & 1 & 0 \\
\hline
\end{array}
\]
Then $\mathbf{B}_5\in\mathsf{BI}(\mathrm{ce})$, and:
\begin{itemize}
\item ${\sim}(a\wedge\neg a)=b\neq 0=\circ a$, so 
$\mathbf{B}_5\notin\mathsf{BI}(\mathrm{ciw})$;
\item $0\wedge\neg\neg 0=0\neq 1=\neg\neg 0$, so 
$\mathbf{B}_5\notin\mathsf{BI}(\mathrm{cf})$;
\item $a\wedge\neg a=a\neq 1=\neg(\circ a)$, so 
$\mathbf{B}_5\notin\mathsf{BI}(\mathrm{ci})$;
\item $\circ a=0\neq 1=\neg(a\wedge\neg a)$, so 
$\mathbf{B}_5\notin\mathsf{BI}(\mathrm{cl})$.
\end{itemize}

Conversely, one checks directly from the tables above that $B_1,\dots,B_4\notin BI(ce)$ (failing, respectively, at $x=b$, $x=a$, $x=1$, and $x=a$).
\end{example}

These relations are summarized in Figure~\ref{fig:hasse2}.

\begin{figure}[ht]
\centering
\begin{tikzpicture}[
  every node/.style={font=\small, inner sep=3pt},
  edge/.style={thick, shorten >=2pt, shorten <=2pt}
]
\node (BI)     at ( 0.0,  4.0) {$\mathsf{BI}$};
\node (BIciw)  at (-2.0,  2.5) {$\mathsf{BI}(\mathrm{ciw})$};
\node (BIcf)   at ( 0.5,  2.5) {$\mathsf{BI}(\mathrm{cf})$};
\node (BIce)   at ( 2.5,  2.5) {$\mathsf{BI}(\mathrm{ce})$};
\node (BIci)   at (-3.0,  1.0) {$\mathsf{BI}(\mathrm{ci})$};
\node (BIcl)   at (-1.0,  1.0) {$\mathsf{BI}(\mathrm{cl})$};
\node (BIcicf) at (-1.0, -0.5)
  {$\mathsf{BI}(\mathrm{ci},\mathrm{cf})=\mathsf{BI}(\mathrm{cl},\mathrm{cf})$};
\draw[edge] (BI)    -- (BIciw);
\draw[edge] (BI)    -- (BIcf);
\draw[edge] (BI)    -- (BIce);
\draw[edge] (BIciw) -- (BIci);
\draw[edge] (BIciw) -- (BIcl);
\draw[edge] (BIci)  -- (BIcicf);
\draw[edge] (BIcl)  -- (BIcicf);
\draw[edge] (BIcf)  -- (BIcicf);
\end{tikzpicture}
\caption{Hasse diagram of the principal subvarieties of $\mathsf{BI}$ 
generated by $(\mathrm{ciw})$, $(\mathrm{ci})$, $(\mathrm{cl})$, 
$(\mathrm{cf})$, and $(\mathrm{ce})$. The varieties 
$\mathsf{BI}(\mathrm{ca}_\#)$ are omitted for clarity.}
\label{fig:hasse2}
\end{figure}
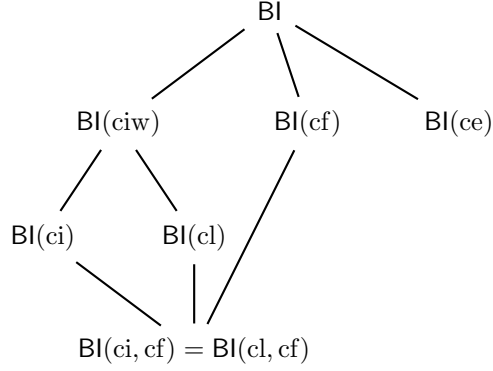

The two diagrams are connected at $\mathsf{BI}(\mathrm{ciw})$, which 
appears as the infimum of the upper diamond in Figure~\ref{fig:hasse1} 
and as a proper subvariety of $\mathsf{BI}$ in Figure~\ref{fig:hasse2}. 
The interaction of these varieties with the propagation axioms 
$(\mathrm{ca}_\#)$, $(\mathrm{a}_\circ)$, and $(\mathrm{a}_\neg)$ will 
be analyzed in Section~\ref{sec:paraconsistency}, where the explosive 
combinations are classified.
\subsection{Paraconsistency under axiomatic extensions}\label{sec:paraconsistency}

We now analyze which combinations of the axioms introduced in 
Section~\ref{sec:extensions} preserve paraconsistency and which force 
the logic to collapse to classical behavior. The main result of this 
subsection is a complete classification of the explosive pairs among 
the fourteen axioms, together with the algebraic mechanism behind 
each case.

We begin with three technical lemmas that will be used throughout 
the proofs.

\begin{lemma}\label{lem:ca-inclusion}
$\mathsf{BI}(\mathrm{ca}_\to)\subseteq\mathsf{BI}(\mathrm{ca}_\vee)$.
\end{lemma}

\begin{proof}
In every Boolean algebra, $x\vee y=(x\to y)\to y$. Hence, for every 
$\mathbf{A}\in\mathsf{BI}(\mathrm{ca}_\to)$ and every $a,b\in A$:
\[
\circ a\wedge\circ b
\leq \circ(a\to b)\wedge\circ b
\leq \circ((a\to b)\to b)
= \circ(a\vee b),
\]
where the first inequality uses $(\mathrm{ca}_\to)$ applied to $a$ 
and $b$, and the second applies $(\mathrm{ca}_\to)$ to $a\to b$ and 
$b$. Thus $(\circ a\wedge\circ b)\wedge\circ(a\vee b)=\circ a\wedge
\circ b$, so $\mathbf{A}\in\mathsf{BI}(\mathrm{ca}_\vee)$.
\end{proof}

\begin{lemma}\label{lem:cl-basic}
In every $\mathbf{A}\in\mathsf{BI}(\mathrm{cl})$, the following hold 
for all $a\in A$:
\begin{enumerate}
\item[(i)] $(a\wedge\neg a)\wedge\neg(a\wedge\neg a)=0$;
\item[(ii)] $\neg(a\wedge\neg a)={\sim}(a\wedge\neg a)$.
\end{enumerate}
\end{lemma}

\begin{proof}
For~(i): substituting the algebraic form of $(\mathrm{cl})$, namely 
$\circ a=\neg(a\wedge\neg a)$, into the BALFI equation 
$a\wedge\neg a\wedge\circ a=0$ gives 
$(a\wedge\neg a)\wedge\neg(a\wedge\neg a)=0$.

For~(ii): let $t=a\wedge\neg a$. By~(i), $t\wedge\neg t=0$. Since 
every BALFI satisfies $t\vee\neg t=1$, the element $\neg t$ is both 
a complement and a Boolean complement of $t$; by uniqueness of 
complements in Boolean algebras, $\neg t={\sim}t$, i.e., 
$\neg(a\wedge\neg a)={\sim}(a\wedge\neg a)$.
\end{proof}

\begin{lemma}\label{lem:cl-cav}
In every $\mathbf{A}\in\mathsf{BI}(\mathrm{cl},\mathrm{ca}_\vee)$, 
the following inequality holds for all $x,y\in A$:
\[
(x\vee y)\wedge\neg(x\vee y)\leq(x\wedge\neg x)\vee(y\wedge\neg y).
\]
\end{lemma}

\begin{proof}
By the algebraic form of $(\mathrm{cl})$ and Lemma~\ref{lem:cl-basic}(ii), 
$\circ x=\neg(x\wedge\neg x)={\sim}(x\wedge\neg x)$ and 
$\circ y={\sim}(y\wedge\neg y)$. Applying $(\mathrm{ca}_\vee)$:
\[
{\sim}(x\wedge\neg x)\wedge{\sim}(y\wedge\neg y)
\leq\circ(x\vee y)
={\sim}((x\vee y)\wedge\neg(x\vee y)).
\]
Taking Boolean complements of both sides and applying De Morgan's law:
\[
(x\vee y)\wedge\neg(x\vee y)
\leq{\sim}{\sim}(x\wedge\neg x)\vee{\sim}{\sim}(y\wedge\neg y)
=(x\wedge\neg x)\vee(y\wedge\neg y).\qedhere
\]

\end{proof}
\subsubsection*{Explosive pairs: the two base cases}

The following two theorems establish the explosivity of the two base 
cases from which all remaining explosive pairs are derived. They are 
proved by distinct algebraic arguments.

\begin{theorem}\label{thm:cl-cav}
Every $\mathbf{A}\in\mathsf{BI}(\mathrm{cl},\mathrm{ca}_\vee)$ 
satisfies $a\wedge\neg a=0$ for all $a\in A$.
\end{theorem}

\begin{proof}
Let $t:=a\wedge\neg a$. By definition, $t\leq a$, so $a\land t=t$.

By Lemma~\ref{lem:cl-basic}, $t\wedge\neg t=0$ and $\neg t={\sim}t$.

We decompose $a$ using the Boolean complement of $t$ and use distributivity, to obtain:
\[
a = a\wedge(t\vee{\sim}t) = t\vee(a\wedge{\sim}t).\tag{$\otimes$}
\]

Applying Lemma~\ref{lem:cl-cav} with $a:=t$ and $b:=a\wedge{\sim}t$,
\[
(t\vee(a\wedge{\sim}t))\wedge\neg(t\vee(a\wedge{\sim}t))
\leq (t\wedge\neg t)\vee((a\wedge{\sim}t)\wedge\neg(a\wedge{\sim}t)).
\]

Since $t\wedge\neg t=0$ and the left-hand side equals $a\wedge\neg a$ (by $\otimes$), 
this reduces to:
\[
a\wedge\neg a \leq (a\wedge{\sim}t)\wedge\neg(a\wedge{\sim}t). 
\tag{$*$}
\]

On the other hand, by definition of meet:
\[
(a\wedge{\sim}t)\wedge\neg(a\wedge{\sim}t)
\leq a\wedge{\sim}t \leq {\sim}t.
\]

Combining with~($*$):
\[
a\wedge\neg a \leq {\sim}t = {\sim}(a\wedge\neg a).
\]

Therefore $a\wedge\neg a=(a\wedge\neg a)\wedge{\sim}(a\wedge\neg a) = 0$.
\end{proof}

\begin{theorem}\label{thm:d1-cato}
Every $\mathbf{A}\in\mathsf{BI}(d1,\mathrm{ca}_\to)$ satisfies $\circ a=1$ and
$a\wedge\neg a=0$, for all $a\in A$.
\end{theorem}

\begin{proof}
Let $\mathbf{A}\in\mathsf{BI}(d1,\mathrm{ca}_\to)$. We first show 
that $\circ a=1$ for every $a\in A$.

Applying $(d1)$ at $a=0$: $\circ 0\vee 0=1$, hence $\circ 0=1$.

Fix an arbitrary $a\in A$. Applying $(\mathrm{ca}_\to)$ with $b=0$:
\[
\circ a\wedge\circ 0\leq\circ(a\to 0).
\]
Since $\circ 0=1$ and $a\to 0={\sim}a$ in any Boolean algebra, 
this gives $\circ a\leq\circ({\sim}a)$.

Applying $(\mathrm{ca}_\to)$ again with ${\sim}a$ and $b=0$:
\[
\circ({\sim}a)\wedge\circ 0\leq\circ({\sim}a\to 0).
\]
Since ${\sim}a\to 0={\sim}({\sim}a)=a$, this gives 
$\circ({\sim}a)\leq\circ a$. Together with the previous inequality:
\[
\circ a=\circ({\sim}a).
\]

Now, apply $(d1)$ to $a$ and to ${\sim}a$:
\[
\circ a\vee a=1 \qquad\text{and}\qquad \circ({\sim}a)\vee{\sim}a=1.
\]
Using $\circ a=\circ({\sim}a)$:
\[
{\sim}{\circ} a\leq a \qquad\text{and}\qquad {\sim}{\circ} a\leq{\sim}a,
\]
hence ${\sim}{\circ} a\leq a\wedge{\sim}a=0$. From this, ${\sim}{\circ} a=0$ and so $\circ a=1$, for every $a$.

Since $\mathbf{A}$ is a BALFI, $a\wedge\neg a\wedge\circ a=0$. As 
$\circ a=1$, we conclude $a\wedge\neg a=0$, for every $a$.
\end{proof}

\subsubsection*{A structural coincidence}

The following result shows that the explosive pair $(\mathrm{cl},\mathrm{ca}_\vee)$ and the pair $(\mathrm{ciw},\mathrm{ca}_\to)$ generate the same variety, namely the variety of Boolean algebras expanded with a constant operator  $\circ a=1$ for all $a$.

\begin{proposition}\label{prop:coincidence}
$\mathsf{BI}(\mathrm{cl},\mathrm{ca}_\vee)=
\mathsf{BI}(\mathrm{ciw},\mathrm{ca}_\to)$.
\end{proposition}

\begin{proof}
We show that both varieties consist precisely of those BALFIs 
satisfying $\neg a={\sim}a$ and $\circ a=1$ for all $a$.

$(\subseteq)$ Let $\mathbf{A}\in\mathsf{BI}(\mathrm{cl},\mathrm{ca}_\vee)$. 
By Theorem~\ref{thm:cl-cav}, $a\wedge\neg a=0$ for all $a$. Together 
with $a\vee\neg a=1$, this gives $\neg a={\sim}a$. Then 
$\circ a=\neg(a\wedge\neg a)=\neg 0=1$. Hence $(\mathrm{ciw})$ holds 
since $\circ a=1={\sim}0={\sim}(a\wedge\neg a)$, and $(\mathrm{ca}_\to)$ 
holds since $(\circ a\wedge\circ b)\wedge\circ(a\to b)=1=
\circ a\wedge\circ b$. Thus 
$\mathbf{A}\in\mathsf{BI}(\mathrm{ciw},\mathrm{ca}_\to)$.

$(\supseteq)$ Let $\mathbf{A}\in\mathsf{BI}(\mathrm{ciw},\mathrm{ca}_\to)$. 
Since $\mathsf{BI}(\mathrm{ciw},\mathrm{ca}_\to)\subseteq
\mathsf{BI}(d1,\mathrm{ca}_\to)$, then $a\wedge\neg a=0$ for all $a$ in $A$, by Theorem~\ref{thm:d1-cato}. 
As proved above, $\neg a={\sim}a$. In turn, $(\mathrm{ciw})$ implies that $\circ a={\sim}(a\wedge\neg a)={\sim}0=1$. 
 Hence $(\mathrm{cl})$ holds since 
$\circ a=1={\sim}0=\neg 0=\neg(a\wedge\neg a)$, and $(\mathrm{ca}_\vee)$ 
holds since $(\circ a\wedge\circ b)\wedge\circ(a\vee b)=1=
\circ a\wedge\circ b$. Thus 
$\mathbf{A}\in\mathsf{BI}(\mathrm{cl},\mathrm{ca}_\vee)$.
\end{proof}

\begin{remark}
In both cases the algebraic collapse forces $\neg$ to coincide with 
the Boolean complement and $\circ$ to be identically $1$, which is 
precisely the classical configuration. 
\end{remark}

\subsubsection*{Explosive pairs: derived cases}

The remaining four explosive pairs now follow by inclusion of 
varieties, using the relations established in 
Section~\ref{sec:varieties} and Proposition~\ref{prop:coincidence}.

\begin{corollary}\label{cor:explosive-pairs}
The following varieties are explosive:\\[1mm]

$\begin{array}{ll}
1. \ \mathsf{BI}(\mathrm{cl},\mathrm{ca}_\to) & \hspace{2cm} 3. \ \mathsf{BI}(\mathrm{ci},\mathrm{ca}_\to)\\
2. \ \mathsf{BI}(\mathrm{ciw},\mathrm{ca}_\to) & \hspace{2cm} 4. \ \mathsf{BI}(i1,\mathrm{ca}_\to).
\end{array}$

\end{corollary}

\begin{proof}
For~(1): by Lemma~\ref{lem:ca-inclusion}, 
$\mathsf{BI}(\mathrm{cl},\mathrm{ca}_\to)\subseteq
\mathsf{BI}(\mathrm{cl},\mathrm{ca}_\vee)$,
so every BALFI in $\mathsf{BI}(\mathrm{cl},\mathrm{ca}_\to)$ 
satisfies $a\wedge\neg a=0$ by Theorem~\ref{thm:cl-cav}.

For~(2)--(4): by Section~\ref{sec:varieties},
\[
\mathsf{BI}(\mathrm{ci})\subseteq
\mathsf{BI}(\mathrm{ciw})\subseteq
\mathsf{BI}(d1),
\qquad
\mathsf{BI}(i1)\subseteq\mathsf{BI}(d1).
\]
Hence:
\[
\mathsf{BI}(\mathrm{ci},\mathrm{ca}_\to)\subseteq
\mathsf{BI}(\mathrm{ciw},\mathrm{ca}_\to)\subseteq
\mathsf{BI}(d1,\mathrm{ca}_\to),
\qquad
\mathsf{BI}(i1,\mathrm{ca}_\to)\subseteq
\mathsf{BI}(d1,\mathrm{ca}_\to).
\]
Every BALFI in each of these varieties satisfies $a\wedge\neg a=0$ 
by Theorem~\ref{thm:d1-cato}.
\end{proof}
Let $K$ denote the class of BALFIs 
$\mathbf{A}=\langle A,\wedge,\vee,\to,\neg,\circ,0,1\rangle$ such that
\[
\neg a={\sim}a
\qquad\text{and}\qquad
{\circ} a=1
\]
for every $a\in A$.
\begin{remark}\label{rem:K-eCPL}
The class $K$ is the algebraic counterpart of the {\em extended classical
logic} $e\mathrm{CPL}$ introduced in~\cite[Remark~30]{CarnielliConiglioMarcos2007}. The logic $e\mathrm{CPL}$, defined over signature $\Sigma$, is obtained from classical
propositional logic by adding axiom \mbox{\em(ext) \hspace{3mm}$\circ\alpha$}. It is immediate to prove that
\[
\Gamma\vdash_{e\mathrm{CPL}}\varphi
\quad\text{if and only if}\quad
\Gamma\vDash_K\varphi.
\]
Observe that $mbC + (\mathrm{ext})= e\mathrm{CPL} = RmbC(\mathrm{ext})$.
\end{remark}
The following result shows that a single explosive core underlies all six combinations identified in this subsection: the variety of BALFIs in which $\neg$ is the Boolean complement and $\circ$ is identically $1$. The apparent diversity of explosive combinations thus collapses to one algebraic phenomenon.
\begin{proposition}[Unification of the explosive pairs]\label{prop:explosive-unification}
The following varieties all coincide with $K$, and hence with one 
another:
\[
\mathsf{BI}(\mathrm{cl},\mathrm{ca}_\vee)=
\mathsf{BI}(\mathrm{ciw},\mathrm{ca}_\to)=
\mathsf{BI}(d1,\mathrm{ca}_\to)=
\mathsf{BI}(\mathrm{cl},\mathrm{ca}_\to)=\mathsf{BI}(\mathrm{ci},\mathrm{ca}_\to)=
\mathsf{BI}(i1,\mathrm{ca}_\to)
=K.
\]
\end{proposition}

\begin{proof}
By Proposition~\ref{prop:coincidence},
\(
\mathsf{BI}(\mathrm{cl},\mathrm{ca}_\vee)
=
\mathsf{BI}(\mathrm{ciw},\mathrm{ca}_\to)
=K.
\)
It therefore remains to show that
\[
\mathsf{BI}(d1,\mathrm{ca}_\to)=
\mathsf{BI}(\mathrm{cl},\mathrm{ca}_\to)=\mathsf{BI}(\mathrm{ci},\mathrm{ca}_\to)=
\mathsf{BI}(i1,\mathrm{ca}_\to)
=K\]

\emph{Each variety is contained in $K$.}
\begin{itemize}
\item $\mathsf{BI}(d1,\mathrm{ca}_\to)\subseteq K$: by 
Theorem~\ref{thm:d1-cato}, $\circ a=1$ and 
$a\wedge\neg a=0$ for every $a$. Hence, since $a\vee\neg a=1$, we have $\neg a={\sim}a$.

\item $\mathsf{BI}(\mathrm{cl},\mathrm{ca}_\to)\subseteq K$: by 
Lemma~\ref{lem:ca-inclusion}, $\mathsf{BI}(\mathrm{cl},\mathrm{ca}_\to)
\subseteq
\mathsf{BI}(\mathrm{cl},\mathrm{ca}_\vee)
=K.$

\item $\mathsf{BI}(\mathrm{ci},\mathrm{ca}_\to)\subseteq K$: by 
Section~\ref{sec:varieties}, $\mathsf{BI}(\mathrm{ci})\subseteq
\mathsf{BI}(\mathrm{ciw})$, and hence $\mathsf{BI}(\mathrm{ci},\mathrm{ca}_\to)
\subseteq
\mathsf{BI}(\mathrm{ciw},\mathrm{ca}_\to)
=K.$

\item $\mathsf{BI}(i1,\mathrm{ca}_\to)\subseteq K$: by 
Section~\ref{sec:varieties}, $\mathsf{BI}(i1)\subseteq\mathsf{BI}(d1)$, and hence $\mathsf{BI}(i1,\mathrm{ca}_\to)
\subseteq
\mathsf{BI}(d1,\mathrm{ca}_\to)
\subseteq K$.
\end{itemize}

\emph{$K$ is contained in each of the remaining four varieties.} Let $\mathbf{A}\in K$, so 
$\neg a={\sim}a$ and $\circ a=1$ for all $a\in A$. We verify each 
axiom directly.
\begin{itemize}
\item $(\mathrm{ca}_\to)$: 
$(\circ a\wedge\circ b)\wedge\circ(a\to b)
=1=\circ a\wedge\circ b$.

\item $(d1)$:
$\circ a\vee a=1\vee a=1$.

\item $(\mathrm{cl})$: $\neg(a\wedge\neg a)={\sim}(a\wedge{\sim}a) ={\sim}0=1=\circ a$.

\item $(\mathrm{ci})$: $\neg(\circ a) ={\sim}(\circ a)={\sim}1=0=a\wedge{\sim}a =a\wedge\neg a$.

\item $(i1)$: $\neg(\circ a)\wedge a={\sim}1\wedge a=0={\sim}1=\neg(\circ a)$.
\end{itemize}

Thus $K$ is contained in each of the four remaining varieties. Together
with Proposition~\ref{prop:coincidence}, this proves that all six
varieties coincide with $K$.
\end{proof}

\begin{remark} The six explosive pairs mentioned above are $(cl,ca_\vee)$, 
$(d1,ca_\to)$, 
$(cl,ca_\to)$,
$(ciw,ca_\to)$,
$(ci,ca_\to)$, and
$(i1,ca_\to)$. For every other pair of principal axioms, except
$(\mathrm{ce},\mathrm{cf})$, we found a finite paraconsistent BALFI
satisfying both axioms by means of a computer-assisted finite-model
search with Mace4 and Python. The same search also confirms that none
of the fourteen individual principal axioms is explosive. The pair
$(\mathrm{ce},\mathrm{cf})$ is exceptional: it does not belong to the
explosive class described above, but it admits no finite
paraconsistent BALFI. This different phenomenon is analyzed separately
below.
\end{remark}
\begin{remark}\label{rem:K-all-axioms}
Every $\mathbf A\in K$ satisfies all fourteen principal axioms, not only the six identified in Proposition~\ref{prop:explosive-unification}: substituting $\neg a={\sim}a$, $\circ a=1$ verifies each remaining axiom directly, e.g. $(\mathrm{ce})$ becomes $a\wedge{\sim}{\sim}a=a$, and $(\mathrm{a}_\circ)$ becomes $\circ(\circ a)=\circ(1)=1$, $\circ$ being constant.
\end{remark}

\begin{corollary}\label{cor:explosive-fmp} Let $Ax$ be any nonempty subset of the fourteen principal axioms containing one of the six minimal explosive pairs $E$. Then $\mathsf{BI}(Ax)=K$. 
\end{corollary}
\begin{proof}
    Let $E\subseteq Ax$ be one of the six minimal explosive pairs. By Proposition~\ref{prop:explosive-unification} $\mathsf{BI}(E)=K$. Hence
    \[\mathsf{BI}(Ax)\subseteq\mathsf{BI}(E)=K\]
Conversely, by Remark~\ref{rem:K-all-axioms}, every algebra in $K$
satisfies all fourteen principal axioms, and therefore
$K\subseteq\mathsf{BI}(Ax)$. Thus $\mathsf{BI}(Ax)=K$.

\end{proof}

    \begin{remark}
In particular, if $Ax_1$ and $Ax_2$ each contain one of the six
minimal explosive pairs, then
$\mathsf{BI}(Ax_1)=\mathsf{BI}(Ax_2)=K$ and, by completeness,
$\mathrm{RmbC}(Ax_1)=\mathrm{RmbC}(Ax_2)$.
\end{remark}
\begin{remark}\label{rem:explosive-fmp}
Every algebra in $K$ is obtained from a Boolean algebra by interpreting
$\neg$ as Boolean complementation and $\circ$ as the constant $1$.
Hence validity in $K$ coincides with classical propositional validity
under this fixed interpretation of $\neg$ and $\circ$. Therefore, every
non-valid formula has a finite Boolean countermodel, and
$\mathrm{RmbC}(Ax)$ has the finite model property for every $Ax$
containing one of the six minimal explosive pairs.
\end{remark}

The six explosive pairs above collapse the algebra through interaction
between consistency and propagation. By contrast, the pair
$(\mathrm{ce},\mathrm{cf})$, to be analyzed in the sequel, imposes a structural
constraint directly on the negation operator.

\subsubsection*{A structural obstruction: the case of 
$(\mathrm{ce},\mathrm{cf})$}

The explosive pairs identified above all involve a propagation axiom 
$(\mathrm{ca}_\#)$. A qualitatively different phenomenon arises for 
the combination $(\mathrm{ce},\mathrm{cf})$: rather than 
forcing explosion through the consistency operator, this combination 
makes the paraconsistent negation involutive, which in the finite 
setting collapses it to the Boolean complement. The following theorem 
identifies this as a structural boundary independent of the 
propagation axioms.

\begin{theorem}\label{thm:involutive}
Let $\mathbf{A}$ be a nontrivial finite Boolean algebra and let 
$\neg:A\to A$ be an operator satisfying, for every $x\in A$,
\[
\neg\neg x=x \qquad\text{and}\qquad x\vee\neg x=1.
\]
Then $\neg={\sim}$, i.e., $\neg$ coincides with the Boolean complement.
\end{theorem}

\begin{proof}

Let \(A\) be a nontrivial finite Boolean algebra. Then \(A\) is isomorphic to the powerset algebra of its set of atoms, that is, \(A\approx\wp (At(A))\). 

Let $h:A\to\wp(At(A))$ be an isomorphism of Boolean algebras,\footnote{For instance, the canonical isomorphism $h(a) = \{x \in At(A) : x \leq a\}$.} and define the transported operation $f:=h\circ\neg\circ h^{-1}$. Since $\neg$ is involutive, $f$ is an involution and hence a bijection. Observe that, for every $a \in A$, $h(\neg a) = h(\neg(h^{-1}(h(a)))) = f(h(a))$. Since $h$ is a Boolean algebra homomorphism, it preserves the top element, i.e., $h(1)=At(A)$. But $1=a \vee \neg a$ for every $a \in A$, then (using again the fact that $h$ is a homomorphism, hence it preserves joins)
\[
At(A) = h(1) = h(a \vee \neg a) = h(a) \cup h(\neg a) = h(a) \cup f(h(a)).\tag{$*$}
\]
Let \(X\in \wp(At(A))\). Since $h$ is a bijection, there exists (a unique) $a \in A$ such that $X=h(a)$, and then $X\cup f(X)=At(A)$, by $(*)$. From this, \(|X\cup f(X)|=|At(A)|\), hence
\[
|X| + |f(X)| - |X\cap f(X)| = |At(A)|.
\]
Let \(n\) be the number of atoms of \(A\). Then,
\[
|X| + |f(X)| = n + |X\cap f(X)|.
\]
Summing this equality over all subsets \(X\in \wp(At(A))\), we get
\[
\sum_{X\in \wp(At(A))} |X| \;+\; \sum_{X\in \wp(At(A))} |f(X)| \;=\; \sum_{X\in \wp(At(A))} \big(n + |X\cap f(X)|\big).
\]
Since \(f\) is a bijection, \(\sum_{X\in \wp(At(A))} |f(X)| = \sum_{X\in \wp(At(A))} |X|\).
Moreover, \(|\wp(At(A))|=2^n\), and therefore
\[
2\sum_{X\in\wp(At(A)) } |X| \;=\; 2^n n \;+\; \sum_{X\in \wp(At(A))} |X\cap f(X)|.
\]
But it is known that
\(\sum_{X\in \wp(At(A))} |X| = n 2^{\,n-1}\). Substituting, we get
\[
2\big(n 2^{\,n-1}\big) \;=\; 2^n n \;=\; 2^n n \;+\; \sum_{X\in \wp(At(A))} |X\cap f(X)|.
\]
Hence
\[
\sum_{X\in \wp(At(A))} |X\cap f(X)| = 0.
\]
This implies that, for every \(X\in \wp(At(A))\), we have
\(|X\cap f(X)|=0\). That is,
\[
X\cap f(X)=\varnothing, \ \text{for every }X\in \wp(At(A)).
\]
Given that $X\cup f(X)=At(A)$, we infer that \(f(X)=At(A)\setminus X\) for every \(X\in \wp(At(A))\). 

Since $h$ is a Boolean algebra isomorphism, the Boolean complement in $A$ is transported to set-theoretic complement, i.e., $h({\sim} a)=At(A)\setminus h(a)$.

On the other hand,
\[
h(\neg a)=f(h(a))
          =At(A)\setminus h(a)
          =h({\sim} a).
\]
Since $h$ is injective, $\neg a={\sim} a$ for every $a\in A$. Hence $\neg={\sim}$.

\end{proof}
\begin{corollary}\label{cor:ce-cf}
There is no finite paraconsistent BALFI satisfying both $(\mathrm{ce})$ 
and $(\mathrm{cf})$.
\end{corollary}

\begin{proof}
In any BALFI satisfying both $(\mathrm{ce})$ and $(\mathrm{cf})$, the 
operator $\neg$ is involutive: $(\mathrm{cf})$ gives $\neg\neg x\leq x$ 
and $(\mathrm{ce})$ gives $x\leq\neg\neg x$, hence $\neg\neg x=x$. 
Moreover, every BALFI satisfies $x\vee\neg x=1$. By 
Theorem~\ref{thm:involutive}, any finite such BALFI satisfies 
$\neg={\sim}$. But then $x\wedge\neg x=x\wedge{\sim}x=0$ for all $x$, 
so the algebra cannot be paraconsistent.
\end{proof}

\begin{remark}
The obstruction identified in Corollary~\ref{cor:ce-cf} is 
independent of the consistency operator $\circ$: it is a purely 
combinatorial constraint on involutions of finite Boolean algebras, 
and no choice of $\circ$ can circumvent it. This distinguishes the 
$(\mathrm{ce},\mathrm{cf})$ case from the explosive pairs 
above, where the collapse was forced by equational interaction 
between $\circ$ and the propagation axioms. In the present case, 
the finite model property with respect to BALFI semantics cannot 
hold for any paraconsistent extension containing both $(\mathrm{ce})$ 
and $(\mathrm{cf})$, and, if such systems are decidable, their decidability would have to be established by different means.
\end{remark}

\begin{example}\label{ex:infinite-ce-cf}
The negative result of Corollary~\ref{cor:ce-cf} only holds in the finite case: 
there exist infinite paraconsistent BALFIs satisfying both 
$(\mathrm{ce})$ and $(\mathrm{cf})$. The following model, taken 
from~\cite{ConiglioFigalloTesta}, witnesses this.

Let $\mathcal{A}=\wp(\mathbb{Z})$ be the Boolean algebra of 
all subsets of $\mathbb{Z}$. For each $n\in\mathbb{Z}$, let
\[
(-\infty,n]:=\{m\in\mathbb{Z}:m\leq n\}
\qquad\text{and}\qquad
[n,\infty):=\{m\in\mathbb{Z}:m\geq n\}.
\]
Let $I:=\{(-\infty,n]:n\in\mathbb{Z}\}\cup\{[n,\infty):n\in\mathbb{Z}\}$ 
and $C:=\mathcal{A}\setminus I$. Define operators 
$\neg,\circ:\mathcal{A}\to\mathcal{A}$ by:
\[
\neg(-\infty,n]=[n,\infty),
\qquad
\neg[n,\infty)=(-\infty,n],
\qquad
\neg X=\mathbb{Z}\setminus X \text{ for } X\in C,
\]
and $\circ X=\mathbb{Z}\setminus(X\cap\neg X)$ for all $X$.

The resulting structure $\mathbf{B}=\langle\wp(\mathbb{Z}),
\cap,\cup,\to,\neg,\circ,\emptyset,\mathbb{Z}\rangle$ is a BALFI. 
The relevant computations are summarized below:
\[
\begin{array}{c|c|c|c|c}
X & \neg X & X\cap\neg X & \circ X & \neg\neg X \\ 
\hline
(-\infty,n] & [n,\infty) & \{n\} & \mathbb{Z}\setminus\{n\} & 
(-\infty,n] \\[2mm]
[n,\infty) & (-\infty,n] & \{n\} & \mathbb{Z}\setminus\{n\} & 
[n,\infty) \\[2mm]
X\in C & \mathbb{Z}\setminus X & \varnothing & \mathbb{Z} & X
\end{array}
\]
From the table: $\neg\neg X=X$ for all $X$, so both $(\mathrm{ce})$ 
and $(\mathrm{cf})$ hold. Moreover, $X\cap\neg X=\{n\}\neq\emptyset$ 
for every $X\in I$, so $\mathbf{B}$ is paraconsistent (since $I \neq \emptyset$). It can be 
verified directly that $\mathbf{B}$ satisfies all the axioms of 
$\mathrm{RmbC}(\mathrm{ce},\mathrm{cf})$. Moreover, it is a paraconsistent model of  
$\mathrm{RCie}:= \mathrm{RmbC}(\mathrm{ci},\mathrm{ce},\mathrm{cf})= \mathrm{RmbC}(\mathrm{ci},\mathrm{cl},\mathrm{ce},\mathrm{cf})$.
\end{example}

\subsection{RCila: explosion and a paraconsistent refinement}

\label{sec:rcila}

\subsubsection*{Da Costa's hierarchy and the logic Cila}

Newton da Costa introduced in~\cite{daCosta1963} a celebrated hierarchy 
of paraconsistent logics $C_n$ ($n\geq 1$) in which the notion of 
well-behavedness $\alpha^{(n)}$ of a formula $\alpha$ is defined inductively in $C_n$ as follows. Let 
$\alpha^{1}\overset{\mathrm{def}}{=}\neg(\alpha\wedge\neg\alpha)$, 
and $\alpha^{n+1}\overset{\mathrm{def}}{=}(\alpha^{n})^{1}$, for $n \geq 1$. 
Finally, let $\alpha^{(1)}\overset{\mathrm{def}}{=} \alpha^1$ and 
$\alpha^{(n+1)}\overset{\mathrm{def}}{=} \alpha^{(n)} \land \alpha^{n+1}$, for $n \geq 1$.
Note that $\alpha^{(n+1)} = \alpha^1 \land \ldots \land \alpha^{n+1}$, for $n \geq 1$.

The logic $C_1$ is the first and most representative system of the 
hierarchy. It extends positive classical logic with a paraconsistent 
negation and the axiom schemas
\[
(\circ\alpha\wedge\circ\beta)\to\circ(\alpha\wedge\beta),
\qquad 
(\circ\alpha\wedge\circ\beta)\to\circ(\alpha\vee\beta),
\qquad
(\circ\alpha\wedge\circ\beta)\to\circ(\alpha\to\beta),
\]
where $\circ\alpha$ abbreviates $\alpha^{(1)}$, together with the 
principle of gentle explosion~$(\mathrm{bc1})$ and double negation 
elimination~$(\mathrm{cf})$. These propagation axioms express that 
consistency is preserved under the Boolean connectives, and are 
precisely what enables the construction of the full da Costa hierarchy.

The presentation of $C_1$ in the language of LFIs, where $\circ$ is 
taken as a primitive consistency operator rather than as an 
abbreviation, is the logic 
Cila~\cite{CarnielliConiglio2016a, CarnielliConiglioMarcos2007}. 
In our notation:
\[
\mathrm{Cila} = \mathrm{mbC}(\mathrm{ci},\mathrm{cl},\mathrm{cf},
\mathrm{ca}_\wedge,\mathrm{ca}_\vee,\mathrm{ca}_\to).
\]

The logic Cila is of central importance in the LFI landscape for two 
reasons. First, it is an LFI in which the consistency operator 
$\circ$ can be defined (up to logical equivalence) from the other connectives as 
$\circ\alpha\equiv\neg(\alpha\wedge\neg\alpha)$.\footnote{The minimal LFI extending mbC with this property is $\mathrm{mbCcicl}:=\mathrm{mbC}(\mathrm{ci},\mathrm{cl})$, introduced in~\cite{CarnielliConiglio2016a}.} Second, 
the propagation axioms $(\mathrm{ca}_\wedge)$, $(\mathrm{ca}_\vee)$, 
$(\mathrm{ca}_\to)$, which abstract the da Costa axioms for $C_1$, combined with axiom (cf) give Cila a very rich structure: they 
are precisely the conditions that force consistency to propagate through 
all Boolean operations, which is what makes the construction of the 
hierarchy $C_n$ possible. However, as we shall see in the sequel, the system Cila cannot be transformed into a self-extensional one.

\subsubsection*{Every BALFI semantics for RCila is explosive}

The self-extensional version of Cila is:
\[
\mathrm{RCila} = \mathrm{RmbC}(\mathrm{ci},\mathrm{cl},\mathrm{cf},
\mathrm{ca}_\wedge,\mathrm{ca}_\vee,\mathrm{ca}_\to).
\]
The question of whether RCila can be paraconsistent was already analyzed 
in~\cite{CarnielliConiglioFuenmayor2022}, showing that the answer is negative. This fact follows directly from the results of 
Section~\ref{sec:paraconsistency}, as we shall see below. By definition, $\mathsf{BI}(\mathrm{RCila})$ is
$\mathsf{BI}(\mathrm{ci},\mathrm{cl},\mathrm{cf},
\mathrm{ca}_\wedge,\mathrm{ca}_\vee,\mathrm{ca}_\to)$.

\begin{proposition}\label{prop:rcila-explosive}
Every $\mathbf{A}\in\mathsf{BI}(\mathrm{RCila})$ satisfies 
$a\wedge\neg a=0$ for all $a\in A$. In particular, $\mathrm{RCila}$ 
is not paraconsistent.
\end{proposition}

\begin{proof}
The logic $\mathrm{RCila}$ includes the axioms $(\mathrm{ci})$ and 
$(\mathrm{ca}_\to)$. By Corollary~\ref{cor:explosive-pairs}(3), every 
BALFI in $\mathsf{BI}(\mathrm{ci},\mathrm{ca}_\to)$ satisfies 
$a\wedge\neg a=0$ for all $a$. Since $\mathrm{RCila}$ contains both $(\mathrm{ci})$ and
$(\mathrm{ca}_\to)$, its BALFI semantics is contained in
$\mathsf{BI}(\mathrm{ci},\mathrm{ca}_\to)$, and hence the same explosion law holds
in every BALFI model of $\mathrm{RCila}$.
\end{proof}

\begin{remark}
The failure of paraconsistency in RCila is a direct consequence of the 
algebraic classification of Section~\ref{sec:paraconsistency}. This gives a self-contained algebraic 
explanation without invoking the argument via the De Morgan schema 
$(\mathrm{dm})$ used in~\cite[Example~4.31]{CarnielliConiglioFuenmayor2022}.
\end{remark}

\subsubsection*{A paraconsistent $C_1$-like LFI with replacement: 
\texorpdfstring{$\mathrm{RC_1^{i2}}$}{RCila minus}}

The failure of paraconsistency in RCila raises a natural question: 
how much of the axiomatic content of Cila (i.e., $C_1$) can be preserved under 
replacement while maintaining paraconsistency? We now identify such a natural weakening, obtained by removing precisely the axioms 
responsible for explosion while preserving the propagation conditions 
and the remaining structural properties of Cila.

The explosive behavior of RCila is caused by the interaction of 
$(\mathrm{ci})$ with $(\mathrm{ca}_\to)$: the axiom $(\mathrm{ci})$ 
forces $\neg{\circ}\alpha$ to imply both $\alpha$ and $\neg\alpha$ 
simultaneously and, in the presence of $(\mathrm{ca}_\to)$, this 
triggers explosion. The axiom $(\mathrm{ci})$ decomposes into two 
weaker conditions: $(i1)$, asserting $\neg{\circ}\alpha\to\alpha$, and 
$(i2)$, asserting $\neg{\circ}\alpha\to\neg\alpha$. Of these, $(i1)$ 
combined with $(\mathrm{ca}_\to)$ remains explosive 
(Corollary~\ref{cor:explosive-pairs}(4)), while the combination of $(\mathrm{ca}_\to)$ with $(i2)$ is not. We therefore remove $(\mathrm{ci})$ and $(\mathrm{cl})$, retaining $(i2)$ as a weaker condition in place of $(\mathrm{ci})$. The condition $(i2)$ is therefore retained as the non-explosive half of $(\mathrm{ci})$, while $(i1)$ is excluded because its combination
with $(\mathrm{ca}_\to)$ is already explosive. Note that  $(\mathrm{cl})$ also needs to be eliminated since, in the presence of replacement and  $(\mathrm{cf})$, it implies  $(\mathrm{ci})$ and hence the condition $(i1)$.

At the same time, we replace $(\mathrm{cl})$ by the propagation axioms 
$(\mathrm{a}_\circ)$ and $(\mathrm{a}_\neg)$, which are also verified 
by $C_1$ and support the hierarchical construction of well-behavedness, 
without triggering explosion. The resulting system is:
\[
\mathrm{RC_1^{i2}} := \mathrm{RmbC}(\mathrm{cf},i2,\mathrm{ca}_\wedge,
\mathrm{ca}_\vee,\mathrm{ca}_\to,\mathrm{a}_\circ,\mathrm{a}_\neg).
\]

This system retains the full propagation of consistency under the 
Boolean connectives $(\mathrm{ca}_\#)$, the axiom $(i2)$ expressing 
that inconsistency implies the paraconsistent negation of the formula, 
the double negation axiom $(\mathrm{cf})$, and the unary propagation 
axioms $(\mathrm{a}_\circ)$ and $(\mathrm{a}_\neg)$. It avoids the two explosive mechanisms identified above: the pair $(\mathrm{ci},\mathrm{ca}_\to)$ and the pair $(\mathrm{cl},\mathrm{ca}_\vee)$.

\begin{theorem}\label{thm:rcila-minus}
$\mathrm{RC_1^{i2}}$ is paraconsistent.
\end{theorem}

\begin{proof}
We exhibit a finite paraconsistent BALFI satisfying all the axioms 
of $\mathrm{RC_1^{i2}}$. We use the BALFI $\mathbf{B}_4$ of 
Example~\ref{ex:B4}, whose operations over $\wp(\{w_1,w_2\})$ are:
\[
\begin{array}{|c|c|c|}
\hline
x & \neg x & \circ x \\
\hline
0 & 1 & 1 \\
a & 1 & 0 \\
b & 1 & 0 \\
1 & 0 & 1 \\
\hline
\end{array}
\]
We verify each axiom of 
$\mathrm{RC_1^{i2}}$. The validity of  (cf), (i2), $(\mathrm{a}_\circ)$ and $(\mathrm{a}_\neg)$ is immediate.

\begin{comment}

$(\mathrm{cf})$: We need $x\wedge\neg\neg x=\neg\neg x$ for all $x$. 
Since $\neg 0=\neg a=\neg b=1$ and $\neg 1=0$, we compute 
$\neg\neg 0=\neg\neg a=\neg\neg b=0$, $\neg\neg 1=1$. 
Then $x\wedge\neg\neg x=\neg\neg x$ holds in all cases. 

$(\mathrm{i2})$: We need $\neg{\circ} x\wedge\neg x=\neg{\circ} x$ for 
all $x$. Since $\circ 0=\circ 1=1$ and $\circ a=\circ b=0$:
\begin{itemize}
\item $x=0$: $\neg{\circ} 0\wedge\neg 0=0\wedge 1=0=\neg{\circ} 0$. 
\item $x=a$: $\neg{\circ} a\wedge\neg a=1\wedge 1=1=\neg{\circ} a$. 
\item $x=b$: analogous to $x=a$. 
\item $x=1$: $\neg{\circ} 1\wedge\neg 1=0\wedge 0=0=\neg{\circ} 1$. 
\end{itemize}

$(\mathrm{a}_\circ)$: We need $\circ{\circ} x=1$ for all $x$. 
Since $\circ a=\circ b=0$ and $\circ 0=\circ 1=1$: $\circ{\circ} a=\circ{\circ} b=\circ 0=1$ and $\circ{\circ} 0=\circ{\circ} 1=\circ 1=1$. 

$(\mathrm{a}_\neg)$: We need $\circ x\wedge\circ\neg x=
\circ x$ for all $x$:
\begin{itemize}
\item $x=0$: $\circ 0\wedge\circ 1=1\wedge 1=1=\circ 0$. 
\item $x=a$: $\circ a\wedge\circ\neg a=0\wedge\circ 1=0\wedge 1=
0=\circ a$. 
\item $x=b$: analogous. 
\item $x=1$: $\circ 1\wedge\circ\neg 1=1\wedge\circ 0=1\wedge 1=
1=\circ 1$. 
\end{itemize}

\end{comment}

\noindent
-- {\em $(\mathrm{ca}_\#)$ for $\#\in\{\wedge,\vee,\to\}$}: In 
$\mathbf{B}_4$, $\circ x=1$ iff $x\in\{0,1\}$, and $\circ x=0$ otherwise. If $\circ x\wedge\circ y=0$ the equation holds 
trivially. If $\circ x\wedge\circ y=1$ then $x,y\in\{0,1\}$; since 
$\{0,1\}$ is closed under $\wedge$, $\vee$, and $\to$ in any Boolean 
algebra, $x\#y\in\{0,1\}$ and hence $\circ(x\#y)=1$, so the equation 
holds. 

\noindent
-- {\em Paraconsistency}: $a\wedge\neg a=a\wedge 1=a\neq 0$. 
\end{proof}

\begin{remark}
The axiom $(i1)$ fails in $\mathbf{B}_4$: 
$\neg{\circ} a\wedge a=1\wedge a=a\neq 1=\neg{\circ} a$, confirming that 
$(i1)$ cannot be added to $\mathrm{RC_1^{i2}}$ without losing 
paraconsistency. Indeed, $(i1)$ together with $(\mathrm{ca}_\to)$ forms one of the
explosive pairs identified in Corollary~\ref{cor:explosive-pairs}.
\end{remark}

%\begin{remark}
The relevance of $\mathrm{RC_1^{i2}}$ is evident. It is well known that $C_1$, and every logic in da 
Costa's hierarchy weaker than it, admits no non-trivial quotient 
algebra in the sense of Lindenbaum--Tarski, and is therefore not 
algebraizable in the general sense of 
Blok--Pigozzi~\cite{BlokPigozzi1989}. $\mathrm{RC_1^{i2}}$ 
retains the full propagation of consistency under the Boolean 
connectives, the very feature that makes $C_1$ the richest and most 
structured system in da Costa's hierarchy, while remaining 
self-extensional and paraconsistent. 
This suggests that an analogue of da Costa's hierarchy could be developed from the outset within the algebraizable, self-extensional setting studied here, rather than retrofitted onto systems that provably resist it. Building such a hierarchy is a natural direction for future work.
%This suggests that an analogue of da Costa's hierarchy could be developed from the outset 
%within the algebraizable, self-extensional setting studied here, rather than retrofitted onto systems that provably resist it. 
%Whether such a hierarchy can be built, and how far it can be pushed while preserving paraconsistency, is a natural direction for future work.
%\end{remark}

\section{Finite Model Property and Decidability of $\mathrm{RmbC}$}\label{sec:fmp}

In this section, we establish the finite model property for $\mathrm{RmbC}$ by means of an algebraic filtration-style construction. The key point is that the variety BI of BALFIs is not locally finite, so the result cannot be obtained by a general local-finiteness argument. Instead, given a non-derivable formula, we isolate the finite set of values taken by its subformulas in a countermodel, close that set under the immediate images of the non-classical operators $\neg$ and $\circ$, and generate the corresponding finite BALFI. 

\begin{proposition}[The variety $\mathsf{BI}$ is not locally finite]\label{prop:non-locally-finite}
The variety of Boolean algebras with LFI operators is not locally finite.
\end{proposition}

\begin{proof}
Consider the algebra
\[
\mathbf{A}=\langle \wp(\mathbb N),\cap,\cup,\rightarrow,\neg,\circ,\emptyset,\mathbb N\rangle,
\]
where the Boolean implication is the usual one, \(\neg X=\mathbb N\setminus X\), and
\[
\circ X=\{\,n+1 : n\in X\,\}
\]
for every \(X\subseteq \mathbb N\).

Since \(\neg\) is the ordinary Boolean complement, we have for every \(X\subseteq \mathbb N\):
\[
X\vee \neg X=\mathbb N
\qquad\text{and}\qquad
X\wedge \neg X\wedge \circ X=\emptyset.
\]
Hence \(\mathbf{A}\) is a BALFI, so \(\mathbf{A}\in\mathsf{BI}\).

Now let \(a=\{0\}\). We write \(\circ^n\) for the \(n\)-fold iteration of \(\circ\), defined by \(\circ^0(a)=a\) and \(\circ^{n+1}(a)=\circ(\circ^n(a))\). A straightforward induction on \(n\) shows that, for every \(n\in\mathbb N\),
\[
\circ^n(a)=\{n\}.
\]
These elements are pairwise distinct. Hence, the subalgebra of \(\mathbf{A}\) generated by the finite set \(\{a\}\) contains infinitely many distinct elements, and is therefore infinite.

Since \(\mathbf{A}\in \mathsf{BI}\) and has a finitely generated subalgebra that is infinite, it follows that \(\mathsf{BI}\) is not locally finite.
\end{proof}

\noindent
This shows that the decidability of $\mathrm{RmbC}$ cannot be obtained from the local finiteness of the BALFI variety alone. 
\begin{theorem}[Finite Model Property for $\mathrm{RmbC}$]\label{thm:fmp-rmbc}
If $\nvdash_{\mathrm{RmbC}}\varphi$, then there exist a finite BALFI $\mathbf{A}$ and a valuation
$v:\mathrm{For}(\Sigma)\to A$ such that $v(\varphi)\neq 1$.
\end{theorem}

\begin{proof}
Assume that $\nvdash_{\mathrm{RmbC}}\varphi$. By the soundness and completeness of $\mathrm{RmbC}$ with respect to BALFIs, there exist a BALFI $\mathbf{B}$ and a valuation
$v_B:\mathrm{For}(\Sigma)\to B$ such that $v_B(\varphi)\neq 1$.

Let $\Phi$ be the set of all subformulas of $\varphi$. By definition, $\Phi$ is finite and closed under subformulas. Define
\[
S=\{v_B(\psi):\psi\in\Phi\}.
\]
and the set
\[
G := S \cup \neg_B[S] \cup \circ_B[S],
\]
where \(\neg_B[S] = \{\neg_B\ a : a \in S\}\) and \(\circ_B[S] = \{\circ_B\ a : a \in S\}\).\\

Since $\Phi$ is finite, $S$ is finite, hence $G$ is finite as well. Let $A$ be the Boolean subalgebra of the Boolean reduct of $\mathbf{B}$ generated by $G$. Since the variety of Boolean algebras is locally finite, $A$ is finite.

We define the non-classical operators on $A$ as follows. If $a\in S$, then we set
\[
\neg_A a := \neg_B a
\qquad\text{and}\qquad
\circ_A a := \circ_B a.
\]
If $a\in A\setminus S$, we set
\[
\neg_A a := {\sim} a
\qquad\text{and}\qquad
\circ_A a := 1,
\]
where ${\sim} a$ is the Boolean complement of $a$ in $A$.\\[2mm]
%\begin{lemma}\label{LS31}
{\bf Fact 1:} The algebra $\mathbf{A} = \langle A,\wedge,\vee,\to,\neg_A,\circ_A,0,1\rangle$
is a finite BALFI.\\[1mm]
%\end{lemma}
%\begin{proof}
{\em Proof of Fact 1:}
Since $A$ is a Boolean subalgebra of the Boolean reduct of $B$, the reduct $\langle A,\wedge,\vee,\to,0,1\rangle$ is a Boolean algebra.

It remains to verify the BALFI equations. The definition of the operations ensures that the original interpretation is preserved on $S$, while outside $S$ the operations are defined so as to trivially satisfy the BALFI equations. So, let $a \in A$.

If $a \in S$, then by definition $\neg_A a = \neg_B a$ and $\circ_A a = \circ_B a$. Since $B$ is a BALFI, we have
\[
a \vee \neg_B a = 1 \quad \text{and} \quad a \wedge \neg_B a \wedge \circ_B a = 0.
\]
Hence the same equations hold in $A$.

If $a \notin S$, then by definition $\neg_A a = {\sim} a$ and $\circ_A a = 1$, where ${\sim} a$ denotes the Boolean complement of $a$ in $A$. Thus,
\[
a \vee \neg_A a = a \vee {\sim} a = 1
\]
and
\[
a \wedge \neg_A a \wedge \circ_A a = a \wedge {\sim} a \wedge 1 = 0.
\]
Therefore, $A$ satisfies the defining equations of a BALFI, proving {\bf Fact 1}.\\
%\end{proof}

Define $v:\mathrm{For}(\Sigma)\to A$ by
\[
v(p)=
\begin{cases}
v_B(p), & \text{if } p \text{ is atomic and } p\in\Phi,\\
0, & \text{if } p \text{ is atomic and } p\notin\Phi,
\end{cases}
\]
and, for compound formulas,
\[
v(\alpha\#\beta)=v(\alpha)\#v(\beta)
\quad\text{for }\#\in\{\wedge,\vee,\rightarrow\},
\qquad
v(\#\alpha)=\#_Av(\alpha)
\quad\text{for }\#\in\{\neg,\circ\}.
\]

%\begin{lemma}\label{lem:preservation}

\noindent
{\bf Fact 2:} $v$ is a valuation over $\mathbf{A}$ such that $v(\psi)=v_B(\psi)$, for every $\psi\in\Phi$.\\[1mm]
%\end{lemma}
%\begin{proof}
{\em Proof of Fact 2:} By definition, $v$ is a valuation over $\mathbf{A}$. For the second part, we argue by structural induction on $\psi$.

If $\psi$ is atomic, the claim follows from the definition of $v$.

If $\psi=\alpha\#\beta$ with $\#\in\{\wedge,\vee,\rightarrow\}$ then, since $\Phi$ is closed under subformulas, $\alpha,\beta\in\Phi$. By the induction hypothesis,
\[
v(\psi)=v(\alpha)\#v(\beta)=v_B(\alpha)\#v_B(\beta)=v_B(\psi).
\]

If $\psi=\#\alpha$  with $\#\in\{\neg,\circ\}$, the proof is also simple: $\alpha\in\Phi$, so $v(\alpha)=v_B(\alpha)\in S$. By definition of $\#_A$ on $S$,
\[
v(\psi)=v(\#\alpha)=\#_A v(\alpha)=\#_A v_B(\alpha)=\#_B v_B(\alpha)=v_B(\#\alpha)=v_B(\psi).
\]
Thus, $v(\psi)=v_B(\psi)$ for every $\psi\in\Phi$. This concludes the proof of {\bf Fact 2}.\\
%\end{proof}

By {\bf Fact 1}, $\mathbf{A}$ is a finite BALFI. By {\bf Fact 2}, for every $\psi \in \Phi$ we have $v(\psi) = v_B(\psi)$. In particular, since $\varphi \in \Phi$, it follows that $v(\varphi) = v_B(\varphi) \neq 1$.
Therefore, $\mathbf{A}$ is a finite BALFI and $v$ is a valuation over it such that $v(\varphi) \neq 1$, so $\varphi$ is not valid in $\mathbf{A}$. This completes the proof.
\end{proof}

\begin{corollary}[Decidability of $\mathrm{RmbC}$]\label{cor:decidability-rmbc}
The logic $\mathrm{RmbC}$ is decidable.
\end{corollary}

\begin{proof}
By Theorem~\ref{thm:fmp-rmbc}, $\mathrm{RmbC}$ has the finite model 
property with respect to BALFIs. Since $\mathrm{RmbC}$ is finitely 
axiomatizable, decidability follows by Theorem~\ref{thm:harrop}.
\end{proof}

\begin{remark}
In particular, theoremhood in \(\mathrm{RmbC}\) can be decided by searching through finite BALFIs of bounded size and checking whether any of them refutes the given formula.
\end{remark}

\section{Transfer of the Finite Model Property to the Main Consistency Extensions}
\label{sec:transfer}
The algebraic filtration argument of Section~\ref{sec:fmp} transfers not only to the fourteen principal axiomatic extensions individually, but to broad classes of their combinations. Every transfer proof follows the same underlying 
construction: isolate a finite support containing the subformula 
values of the target formula, close it appropriately, generate a 
finite Boolean subalgebra, and extend the non-classical operators 
canonically outside the support. We isolate this construction once, 
as a general lemma, and then verify for each axiom only the specific 
equation it requires, the genuinely new content in each case.

\subsection{The general filtration construction}
\label{sec:transfer-general}

\begin{lemma}\label{lem:general-transfer}
Let $\mathbf{B}\in\mathsf{BI}$, let $\varphi\in\mathrm{For}(\Sigma)$ 
with $\Phi=\mathrm{Sub}(\varphi)$, and let $v_B:\mathrm{For}(\Sigma)\to B$ 
be a valuation. Let $S:=\{v_B(\psi):\psi\in\Phi\}$. Let $A$ be a 
finite Boolean subalgebra of the Boolean reduct of $\mathbf{B}$, and 
let $R_\neg\subseteq R_\circ\subseteq A$ satisfy
\[
S\subseteq R_\neg,
\qquad
\neg_B[R_\neg]\subseteq A,
\qquad
\circ_B[R_\circ]\subseteq A.
\]
Fix any function $f_\circ:A\setminus R_\circ\to A$, and define 
$\neg_A,\circ_A:A\to A$ by
\[
\neg_A a=
\begin{cases}
\neg_B a, & a\in R_\neg,\\
{\sim}a, & a\notin R_\neg,
\end{cases}
\qquad
\circ_A a=
\begin{cases}
\circ_B a, & a\in R_\circ,\\
f_\circ(a), & a\notin R_\circ.
\end{cases}
\]
Then, regardless of the choice of $f_\circ$:
\begin{enumerate}
\item[(a)] $\mathbf{A}=\langle A,\wedge,\vee,\to,\neg_A,\circ_A,0,1\rangle$ 
is a finite BALFI;
\item[(b)] the valuation $v:\mathrm{For}(\Sigma)\to A$ over $\mathbf{A}$ defined by 
$v(p)=v_B(p)$ for atomic $p\in\Phi$, $v(p)=0$ otherwise, satisfies $v(\psi)=v_B(\psi)$ for every $\psi\in\Phi$. In particular, $v(\varphi)=v_B(\varphi)$.
\end{enumerate}
\end{lemma}

\begin{proof}
(a) The reduct $\langle A,\wedge,\vee,\to,0,1\rangle$ of $\mathbf{A}$ is a Boolean 
algebra, since $A$ is a Boolean subalgebra of the corresponding reduct of $\mathbf{B}$. We 
verify the two BALFI equations for every $a\in A$, in three cases.\\

\noindent
\emph{Case $a\in R_\neg$.} Since $R_\neg\subseteq R_\circ$, also
$a\in R_\circ$, so $\neg_A a=\neg_B a$ and $\circ_A a=\circ_B a$.
Since $\mathbf{B}$ is a BALFI,
\[
a\vee\neg_B a=1
\qquad\text{and}\qquad
a\wedge\neg_B a\wedge\circ_B a=0.
\]
These are the same equations for $\mathbf{A}$, since $a\in A$ and
$\neg_B a,\circ_B a\in A$ by the closure assumptions.\\

\noindent
\emph{Case $a\in R_\circ\setminus R_\neg$.} Here $\neg_A a={\sim}a$ 
(since $a\notin R_\neg$) and $\circ_A a=\circ_B a$ (since $a\in R_\circ$). 
Then $a\vee\neg_A a=a\vee{\sim}a=1$, and 
$a\wedge\neg_A a\wedge\circ_A a=(a\wedge{\sim}a)\wedge\circ_B a
=0\wedge\circ_B a=0$.\\

\noindent
\emph{Case $a\notin R_\circ$.} Since $R_\neg\subseteq R_\circ$, also 
$a\notin R_\neg$, so $\neg_A a={\sim}a$ and $\circ_A a=f_\circ(a)$. 
Then $a\vee\neg_A a=a\vee{\sim}a=1$, and 
$a\wedge\neg_A a\wedge\circ_A a=(a\wedge{\sim}a)\wedge f_\circ(a)
=0\wedge f_\circ(a)=0$, regardless of the value of $f_\circ(a)$.\\

\noindent
(b) We prove by induction on the structure
of $\psi\in\Phi$ that $v(\psi)=v_B(\psi)$.

If $\psi$ is atomic, then, since $\psi\in\Phi$, \(v(\psi)=v_B(\psi)\) by the definition of $v$.

If $\psi=\alpha\#\beta$, with
$\#\in\{\wedge,\vee,\rightarrow\}$, then $\alpha,\beta\in\Phi$.
By the induction hypothesis,
$v(\alpha)=v_B(\alpha)$ and $v(\beta)=v_B(\beta)$. Since $A$ is a
Boolean subalgebra of the Boolean reduct of $\mathbf{B}$, the
interpretations of $\wedge,\vee,\rightarrow$ agree on $A$, and hence
\[
v(\psi)
=v(\alpha)\#v(\beta)
=v_B(\alpha)\#v_B(\beta)
=v_B(\psi).
\]

If $\psi=\neg\alpha$, then $\alpha\in\Phi$, so
$v(\alpha)=v_B(\alpha)\in S\subseteq R_\neg$. Hence
\[
v(\psi)
=v(\neg\alpha)
=\neg_Av(\alpha)
=\neg_Bv(\alpha)
=v_B(\neg\alpha)
=v_B(\psi).
\]

If $\psi=\circ\alpha$, then similarly
$\alpha\in\Phi$ and $v(\alpha)=v_B(\alpha)\in S\subseteq R_\circ$.
Therefore
\[
v(\psi)
=v(\circ\alpha)
=\circ_Av(\alpha)
=\circ_Bv(\alpha)
=v_B(\circ\alpha)
=v_B(\psi).
\]
\end{proof}

In every application below, $f_\circ$ is a constant function except 
in the case of $(i2)$, which uses the identity $f_\circ(a)=a$; the 
constant used is $f_\circ\equiv 1$ for $(\mathrm{ciw})$, 
$(\mathrm{ci})$, $(\mathrm{cl})$, $(\mathrm{cf})$, $(d1)$, $(d2)$, 
$(\mathrm{a}_\circ)$, $(i1)$, and $(\mathrm{a}_\neg)$, and $f_\circ\equiv 0$ for $(\mathrm{ca}_\#)$; the axiom $(\mathrm{cf})$ is compatible with either default, since its equation does not involve $\circ$, and is combined with each in Theorems~\ref{thm:transfer-basic} and~\ref{thm:transfer-ca}, respectively. Finiteness of $A$ in each application below 
follows because $A$ is generated, as a Boolean subalgebra, by a 
finite set, Boolean algebras being locally finite.

\subsection{Transfer for \texorpdfstring{$\{(\mathrm{ciw}), (\mathrm{ci}),(\mathrm{cl}), (d1),(d2), (\mathrm{a}_\circ),(\mathrm{cf}),(i1),(\mathrm{a}_\neg)\}$}{ciw, ci, cl, cf, d1, d2, a-circ, i1, and a-neg}}
\label{sec:transfer-basic}

We first establish three auxiliary lemmas: the first is used in the proof of $(\mathrm{ci})$, the second in the proof of $(\mathrm{a}_\circ)$, and the third in the proof of $(i1)$.

\begin{lemma}\label{lem:ci-facts}
In every $\mathbf{B}\in\mathsf{BI}(\mathrm{ci})$, one has $\neg 1=0$ 
and $\circ 0=1$.
\end{lemma}

\begin{proof}
By $(\mathrm{ci})$ at $a=0$: $\neg(\circ 0)=0\wedge\neg 0=0$. Now 
apply the BALFI axiom $x\vee\neg x=1$ with $x=\circ 0$: since 
$\neg(\circ 0)=0$, this gives $\circ 0\vee 0=1$, hence $\circ 0=1$. 
Since $\circ 0=1$, replacing $\circ 0$ by $1$ in the equation 
$\neg(\circ 0)=0$ established above yields $\neg 1=0$.
\end{proof}
\begin{lemma}\label{lem:acirc-top}
If $\mathbf{B}\in\mathsf{BI}(\mathrm{a}_\circ)$, then $\circ 1=1$.
\end{lemma}

\begin{proof}
Let $y:=\circ 0$. By $(\mathrm{a}_\circ)$ at $a=0$: $\circ y=1$. 
Applying $(\mathrm{a}_\circ)$ at $a=y$: $\circ(\circ y)=1$. Since 
$\circ y=1$, this gives $\circ 1=1$.
\end{proof}
\begin{lemma}\label{lem:i1-top}
If $\mathbf{B}\in\mathsf{BI}(i1)$, then $\neg 1=0$.
\end{lemma}

\begin{proof}
By $(i1)$ at $a=0$: $\neg(\circ 0)\wedge 0=\neg(\circ 0)$. The 
left-hand side is $0$, so $\neg(\circ 0)=0$. Since 
$\mathsf{BI}(i1)\subseteq\mathsf{BI}(d1)$ 
(Section~\ref{sec:varieties}), $(d1)$ at $a=0$ gives $\circ 0\vee 0
=1$, hence $\circ 0=1$. Substituting, $\neg 1=\neg(\circ 0)=0$.
\end{proof}
\begin{theorem}\label{thm:transfer-basic}
Let $Ax$ be a nonempty subset of $\{(\mathrm{ciw}),(\mathrm{ci}),(\mathrm{cl}),(\mathrm{cf}),(d1),(d2),(\mathrm{a}_\circ),(i1),(\mathrm{a}_\neg)\}$. If $\nvdash_{\mathrm{RmbC}(Ax)}\varphi$, then there exist a finite BALFI $\mathbf{A}\in\mathsf{BI}(Ax)$ and a valuation $v:\mathrm{For}(\Sigma)\to A$ over $\mathbf{A}$ such that $v(\varphi)\neq 1$.
\end{theorem}

\begin{proof}
Let $\Phi=\mathrm{Sub}(\varphi)$. By completeness of $\mathrm{RmbC}(Ax)$
with respect to $\mathsf{BI}(Ax)$, there exist a BALFI
$\mathbf{B}\in\mathsf{BI}(Ax)$ and a valuation
$v_B:\mathrm{For}(\Sigma)\to B$ over $\mathbf{B}$ such that $v_B(\varphi)\neq 1$. Let
$S:=\{v_B(\psi):\psi\in\Phi\}$.

Let $T:=S\cup\neg_B[S]\cup\{{\sim}x:x\in S\}\cup\{{\sim}x:x\in\neg_B[S]\}$,
which is finite and closed under Boolean complementation. Let $A$ be the Boolean subalgebra generated by $T\cup\neg_B[T]\cup\circ_B[T]$. We apply
Lemma~\ref{lem:general-transfer} with $R_\neg=R_\circ=T$ and
$f_\circ(a)=1$ for all $a\notin T$. We verify, for each axiom
belonging to $Ax$, that the resulting algebra satisfies the
corresponding equation.

If $(\mathrm{ciw})\in Ax$, then we verify
$\circ_A a={\sim}(a\wedge\neg_A a)$ for all $a\in A$. 
\begin{itemize}
    \item If $a\in T$,
then $\neg_A a=\neg_B a$ and $\circ_A a=\circ_B a$, so the equation
holds because $\mathbf{B}\in\mathsf{BI}(\mathrm{ciw})$.
\item If $a\notin T$,
then $\neg_A a={\sim}a$ and $\circ_A a=1$, while $a\wedge\neg_A a=0$;
hence ${\sim}(a\wedge\neg_A a)={\sim}0=1=\circ_A a$.
\end{itemize}

If $(\mathrm{cl})\in Ax$, then we verify
$\circ_A a=\neg_A(a\wedge\neg_A a)$ for all $a\in A$.
\begin{itemize}
\item If $a\in T$, put $t:=a\wedge\neg_B a$. By
Lemma~\ref{lem:cl-basic}(ii), $\neg_B t={\sim}t$. Moreover,
$\neg_A a=\neg_B a$, hence $a\wedge\neg_A a=t$. 

If $t\in T$, then
$\neg_A t=\neg_B t={\sim}t$; if $t\notin T$, then $\neg_A t={\sim}t$
by definition. Thus, in either case,
$\neg_A(a\wedge\neg_A a)={\sim}(a\wedge\neg_B a)=\neg_B(a\wedge\neg_B a)
=\circ_B a=\circ_A a$, where the last equality follows from $a\in T$.
\item If $a\notin T$, then $\neg_A a={\sim}a$, so $a\wedge\neg_A a=0$.
Hence $\neg_A(a\wedge\neg_A a)=\neg_A 0=1=\circ_A a$. Indeed, if
$0\in T$, then $\neg_A 0=\neg_B 0=1$, since $0\vee\neg_B 0=1$ in every
BALFI; if $0\notin T$, then $\neg_A 0={\sim}0=1$ by definition. Note that $a\notin T$, hence $\circ_A a = f_\circ(a)=1$.
\end{itemize}

If $(\mathrm{ci})\in Ax$, then we verify
$\neg_A(\circ_A a)=a\wedge\neg_A a$ for all $a\in A$.
\begin{itemize}
\item If $a\in T$, then $\circ_A a=\circ_B a$ and $\neg_A a=\neg_B a$.

If $\circ_B a\in T$, then
$\neg_A(\circ_A a)=\neg_B(\circ_B a)=a\wedge\neg_B a=a\wedge\neg_A a$
because $\mathbf{B}\in \mathsf{BI}(\mathrm{ci})$. 

If $\circ_B a\notin T$, then
$\neg_A(\circ_B a)={\sim}(\circ_B a)$. Since $\mathbf{B}\in \mathsf{BI}(\mathrm{ci})$,
$\neg_B(\circ_B a)=a\wedge\neg_B a$. Thus, using
Lemma~\ref{lem:neg-complement-bound}, ${\sim}(\circ_B a)\leq
\neg_B(\circ_B a)=a\wedge\neg_B a$. Conversely, since
$a\wedge\neg_B a\wedge\circ_B a=0$, we have
$a\wedge\neg_B a\leq{\sim}(\circ_B a)$. Hence
$\neg_A(\circ_A a)={\sim}(\circ_B a)=a\wedge\neg_B a=a\wedge\neg_A a$.
\item If $a\notin T$, then $\neg_A a=\sim a$ and $\circ_A a=1$. Hence $\neg_A(\circ_A a)=\neg_A 1$.
If $1\in T$, Lemma~\ref{lem:ci-facts} gives $\neg_A 1=\neg_B 1=0$; if
$1\notin T$, then $\neg_A 1={\sim}1=0$. Therefore
$\neg_A(\circ_A a)=0=a\wedge{\sim}a=a\wedge\neg_A a$.
\end{itemize}

If $(d1)\in Ax$, then we verify $\circ_A a\vee a=1$ for all $a\in A$.
\begin{itemize}
    \item  If $a\in T$, then $\circ_A a=\circ_B a$. Since $\mathbf{B}\in \mathsf{BI}(d1)$, $\circ_B a\vee a=1$.
\item If $a\notin T$, then $\circ_A a=1$ by definition.
\end{itemize}

If $(d2)\in Ax$, then we verify $ \circ_A a\vee\neg_A a=1$ for all $a\in A$
\begin{itemize}
\item If $a\in T$, then $\circ_A a=\circ_B a$ and $\neg_A a=\neg_B a$. Since $\mathbf{B}\in \mathsf{BI} (d2)$, $\circ_B a\vee\neg_B a=1$. 
\item If
$a\notin T$, then $\circ_A a=1$. Hence $ \circ_A a\vee\neg_A a=1\vee\neg_A a=1$.
\end{itemize}

If $(\mathrm{a}_\circ)\in Ax$, then we check $\circ_A(\circ_A a)=1$
for all $a\in A$. 
\begin{itemize}
    \item If $a\in T$: $\circ_A a=\circ_B a$. If $\circ_B a\in T$:
$\circ_A(\circ_B a)=\circ_B(\circ_B a)=1$, since
$\mathbf{B}\in\mathsf{BI}(\mathrm{a}_\circ)$. If $\circ_B a\notin T$:
$\circ_A(\circ_B a)=1$ by definition. Either way,
$\circ_A(\circ_A a)=1$.
\item If $a\notin T$: $\circ_A a=1$. If $1\in T$:
$\circ_A 1=\circ_B 1=1$ by Lemma~\ref{lem:acirc-top}. If $1\notin T$:
$\circ_A 1=1$ by definition. Either way, $\circ_A(\circ_A a)=1$.
\end{itemize}

If $(\mathrm{cf})\in Ax$, then we verify $\neg_A\neg_A a\leq a$ for
all $a\in A$.
\begin{itemize}
\item If $a\in T$: $\neg_A a=\neg_B a$. If $\neg_B a\in T$:
$\neg_A(\neg_B a)=\neg_B(\neg_B a)\leq a$, because
$\mathbf{B}\in \mathsf{BI}(\mathrm{cf})$.

If $\neg_B a\notin T$:
$\neg_A(\neg_B a)={\sim}(\neg_B a)$; by
Lemma~\ref{lem:neg-complement-bound} applied at $x=\neg_B a$,
${\sim}(\neg_B a)\leq\neg_B(\neg_B a)$, and $\neg_B(\neg_B a)\leq a$
because $\mathbf{B}\in \mathsf{BI}(\mathrm{cf})$. Either way, $\neg_A\neg_A a\leq a$.
\item If $a\notin T$: $\neg_A a={\sim}a$, and since $T$ is closed
under Boolean complement, it follows that ${\sim}a\notin T$ as well. Hence
$\neg_A\neg_A a=\neg_A({\sim}a)={\sim}{\sim}a=a\leq a$.
\end{itemize}

If $(i1)\in Ax$, then we check
$\neg_A(\circ_A a)\wedge a=\neg_A(\circ_A a)$ for all $a\in A$.
\begin{itemize}
\item If $a\in T$: $\circ_A a=\circ_B a$. If $\circ_B a\in T$:
$\neg_A(\circ_B a)=\neg_B(\circ_B a)$, and the equation
$\neg_B(\circ_B a)\wedge a=\neg_B(\circ_B a)$ holds in $\mathbf{B}$, because $\mathbf{B}\in\mathsf{BI}(i1)$. 

If $\circ_B a\notin T$:
$\neg_A(\circ_A a)=\neg_A(\circ_B a)={\sim}(\circ_B a)$. Hence it suffices to prove
${\sim}(\circ_B a)\wedge a={\sim}(\circ_B a)$, which, since $A$ is a Boolean subalgebra of the Boolean reduct of
$\mathbf B$, is equivalent to $\circ_B a\vee a=1.$ This holds because $\mathsf{BI}(i1)\subseteq\mathsf{BI}(d1)$, so $\mathbf B\in\mathsf{BI}(d1)$ and therefore $\circ_B a\vee a=1$.

\item If $a\notin T$: $\circ_A a=1$. If $1\in T$:
$\neg_A 1 =\neg_B 1 =0$ by Lemma~\ref{lem:i1-top}. If $1\notin T$:
$\neg_A 1={\sim}1=0$. Either way, $\neg_A 1\wedge a=0=\neg_A 1 $.
\end{itemize}

If $(\mathrm{a}_\neg)\in Ax$, then we check
$\circ_A a\wedge\circ_A(\neg_A a)=\circ_A a$ for all $a\in A$.
\begin{itemize}
\item If $a\in T$: $\circ_A a=\circ_B a$, $\neg_A a=\neg_B a$. If
$\neg_B a\in T$: $\circ_A(\neg_B a)=\circ_B(\neg_B a)$, and the
equation $\circ_B a\wedge\circ_B(\neg_B a)=\circ_B a$ holds in
$\mathbf{B}$. If $\neg_B a\notin T$: $\circ_A(\neg_B a)=1$ by
default, so $\circ_A a\wedge 1=\circ_A a$ trivially.
\item If $a\notin T$: $\neg_A a={\sim}a$. Since $T$ is closed under
Boolean complementation, ${\sim}a\notin T$. Hence
$\circ_A a=1$ and $\circ_A({\sim}a)=1$, so $1\wedge 1=1$.
\end{itemize}

In every case, Lemma~\ref{lem:general-transfer}(b) gives
$v(\varphi)=v_B(\varphi)\neq 1$.
\end{proof}

\subsection{Transfer for \texorpdfstring{$\{(\mathrm{cf}), (\mathrm{ca}_\#)\}$}
{ca-sharp}, \texorpdfstring{$\#\in\{\wedge,\vee,\to\}$}{\# in \{and,or,implies\}}}
\label{sec:transfer-ca}

The axioms $(\mathrm{ca}_\#)$ involve $\circ$ applied to a Boolean 
combination of two elements, so the region on which $\circ_A$ must be 
correct is strictly larger than the region needed for $\neg_A$.

\begin{theorem}\label{thm:transfer-ca}
Let $Ax$ be a nonempty subset of
$\{(\mathrm{cf}),(\mathrm{ca}_\wedge),(\mathrm{ca}_\vee),(\mathrm{ca}_\to)\}$
such that $Ax\cap
\{(\mathrm{ca}_\wedge),(\mathrm{ca}_\vee),(\mathrm{ca}_\to)\}
\neq\emptyset$. If $\nvdash_{\mathrm{RmbC}(Ax)}\varphi$, then there exist a finite
BALFI $\mathbf{A}\in\mathsf{BI}(Ax)$ and a valuation
$v:\mathrm{For}(\Sigma)\to A$ over $\mathbf{A}$ such that $v(\varphi)\neq 1$.
\end{theorem}

\begin{proof}
Let $\Phi=\mathrm{Sub}(\varphi)$. By completeness of
$\mathrm{RmbC}(Ax)$ with respect to $\mathsf{BI}(Ax)$, there exist
a BALFI $\mathbf{B}\in\mathsf{BI}(Ax)$ and a valuation
$v_B:\mathrm{For}(\Sigma)\to B$ over $\mathbf{B}$ such that $v_B(\varphi)\neq 1$. Let $S:=\{v_B(\psi):\psi\in\Phi\}$.

Let $A_0$ be the Boolean subalgebra generated by
$S\cup\neg_B[S]\cup\circ_B[S]$, and let $A$ be the Boolean subalgebra
generated by
$A_0\cup\neg_B[A_0]\cup\circ_B[A_0]$.
We apply Lemma~\ref{lem:general-transfer} with
$R_\neg=R_\circ=A_0$ and $f_\circ(a)=0$ for all
$a\notin A_0$.

For every $(\mathrm{ca}_\#)\in Ax$, if $a,b\in A_0$, then $a\#b\in A_0$ because $A_0$ is a Boolean subalgebra, and $\circ_A a=\circ_B a, \circ_A b=\circ_B b, \circ_A(a\#b)=\circ_B(a\#b)$. Hence the equation $(\mathrm{ca}_\#)$ holds in $\mathbf A$ because it holds in $\mathbf B$. If $a\notin A_0$ or $b\notin A_0$, then $\circ_A a\wedge\circ_A b=0$, so $(\mathrm{ca}_\#)$ holds trivially.

If $(\mathrm{cf})\in Ax$, we verify
$\neg_A\neg_A a\leq a$ for all $a\in A$.
\begin{itemize}
    \item If $a\in A_0$, then $\neg_Aa=\neg_Ba$. If
$\neg_Ba\in A_0$, then $\neg_A(\neg_Ba)=\neg_B(\neg_Ba)\leq a$ because $\mathbf B\in\mathsf{BI}(\mathrm{cf})$. 

If $\neg_Ba\notin A_0$, then $\neg_A(\neg_Ba)={\sim}(\neg_Ba)
\leq \neg_B(\neg_Ba)\leq a$, where the first inequality follows from Lemma~\ref{lem:neg-complement-bound}.
\item If $a\notin A_0$, then ${\sim}a\notin A_0$, because $A_0$ is a Boolean
subalgebra. Hence
$\neg_Aa={\sim}a
$ and $\neg_A\neg_Aa={\sim}{\sim}a=a.$
Thus $(\mathrm{cf})$ holds in $\mathbf A$.
\end{itemize}
By Lemma~\ref{lem:general-transfer}(b),
$v(\varphi)=v_B(\varphi)\neq 1$.
\end{proof}

\subsection{Transfer for \texorpdfstring{$\{(\mathrm{cf}),(i2)\}$}{cf,i2}}
\label{sec:transfer-i}

The axiom $(i2)$ involves $\neg(\circ\alpha)$, so the support must be
correct not only on $S$ but on the whole subalgebra it generates.

\begin{theorem}\label{thm:transfer-i2}
Let $Ax$ be a nonempty subset of $\{(i2),(\mathrm{cf})\}$ with $(i2)\in Ax$. If $\nvdash_{\mathrm{RmbC}(Ax)}\varphi$, then there exist a finite BALFI $\mathbf{A}\in\mathsf{BI}(Ax)$ and a valuation $v:\mathrm{For}(\Sigma)\to A$ over $\mathbf{A}$ such that $v(\varphi)\neq 1$.
\end{theorem}

\begin{proof}
Let $\Phi=\mathrm{Sub}(\varphi)$. By completeness of $\mathrm{RmbC}(Ax)$
with respect to $\mathsf{BI}(Ax)$, there exist a BALFI
$\mathbf{B}\in\mathsf{BI}(Ax)$ and a valuation
$v_B:\mathrm{For}(\Sigma)\to B$ over $\mathbf{B}$ such that $v_B(\varphi)\neq 1$. Let
$S:=\{v_B(\psi):\psi\in\Phi\}$.

Let $A_0$ be the Boolean subalgebra of $\mathbf{B}$ generated by
$S\cup\neg_B[S]\cup\circ_B[S]$, and let $A$ be the Boolean subalgebra of $\mathbf{B}$
generated by $A_0\cup\neg_B[A_0]\cup\circ_B[A_0]$. Note that $A$ is finite. We apply
Lemma~\ref{lem:general-transfer} with $R_\neg=R_\circ=A_0$ and the
identity default $f_\circ(a)=a$ for all $a\notin A_0$.

If $(\mathrm{i2})\in Ax$, we verify $\neg_A(\circ_A a)\wedge\neg_A a=\neg_A(\circ_A a)$ for all
$a\in A$.
\begin{itemize}
    \item If $a\in A_0$, $\circ_A a=\circ_B a$, $\neg_A a=\neg_B a$. If
$\circ_B a\in A_0$: the equation $\neg_B(\circ_B a)\wedge\neg_B a
=\neg_B(\circ_B a)$ holds in $\mathbf{B}$. 

If $\circ_B a\notin A_0$:
$\neg_A(\circ_A a)=\neg_A(\circ_B a)={\sim}(\circ_B a)$. Hence it suffices to prove
${\sim}(\circ_B a)\wedge \neg_{B}a={\sim}(\circ_B a)$, which, since $A$ is a Boolean subalgebra of the Boolean reduct of
$\mathbf B$, is equivalent to $\circ_B a\vee \neg_{B} a=1.$ This holds because $\mathsf{BI}(i2)\subseteq\mathsf{BI}(d2)$, so $\mathbf B\in\mathsf{BI}(d2)$ and therefore $\circ_B a\vee \neg_{B}a=1$.

\item If $a\notin A_0$, $\circ_A a=a$, so
$\neg_A(\circ_A a)\wedge\neg_A a=\neg_A a\wedge\neg_A a=\neg_A a=\neg_A(\circ_A a)$
by idempotence of $\wedge$.
\end{itemize}

If $(\mathrm{cf})\in Ax$, we verify $\neg_A\neg_A a\leq a$
for all $a\in A$.
\begin{itemize}
\item If $a\in A_0$: $\neg_A a=\neg_B a$. If $\neg_B a\in A_0$:
$\neg_A(\neg_B a)=\neg_B(\neg_B a)\leq a$, because
$\mathbf{B}\in \mathsf{BI}(\mathrm{cf})$. If $\neg_B a\notin A_0$:
$\neg_A(\neg_B a)={\sim}(\neg_B a)\leq\neg_B(\neg_B a)\leq a$, by
Lemma~\ref{lem:neg-complement-bound}.
\item If $a\notin A_0$: since $A_0$ is a Boolean subalgebra,
${\sim} a\notin A_0$.
So $\neg_A a={\sim}a$ and $\neg_A\neg_A a={\sim}{\sim}a=a\leq a$.
\end{itemize}

By Lemma~\ref{lem:general-transfer}(b), $v(\varphi)=v_B(\varphi)\neq 1$.
\end{proof}

\begin{remark}\label{rem:i1-i2-asymmetry}
Although $(i1)$ and $(i2)$ look syntactically parallel, their proofs
are not: the verification of $(i1)$ in Theorem~\ref{thm:transfer-basic}
relies on the fact that $\neg_B1=0$ in $\mathsf{BI}(i1)$
(Lemma~\ref{lem:i1-top}), which fails in $\mathsf{BI}(i2)$. Indeed,
the algebra $\mathbf{B}_5$ from Example~\ref{ex:B5} belongs to
$\mathsf{BI}(i2)$, while $\neg 1=1\neq0$. Consequently, the transfer
for $(i2)$ uses the identity default $f_\circ(a)=a$, which makes the
required equation hold by idempotence.
\end{remark}

\subsection{The case of \texorpdfstring{$(\mathrm{ce})$}{ce}}
\label{sec:transfer-ce}

The axiom $(\mathrm{ce})$ is the only one of the fourteen principal 
axioms for which the transfer of the finite model property is not 
established in this paper. Two distinct observations are relevant 
here, of different scopes and of different logical strength.

\subsubsection*{The combination \texorpdfstring{$(\mathrm{ce},
\mathrm{cf})$}{\{ce,cf\}}: a definitive impossibility}

When $(\mathrm{ce})$ is combined with $(\mathrm{cf})$, 
Corollary~\ref{cor:ce-cf} already settles the matter completely: no 
finite BALFI can satisfy both axioms while remaining paraconsistent. 
This is not a limitation of the filtration technique but a genuine 
mathematical impossibility, independent of any particular method.

\subsubsection*{The axiom \texorpdfstring{$(\mathrm{ce})$}{ce} alone: 
a methodological obstruction}

The situation for $(\mathrm{ce})$ \emph{by itself}, without 
$(\mathrm{cf})$, is different in kind. The axiom $(\mathrm{ce})$ alone 
only yields the inequality $a\leq\neg\neg a$, not the equality 
$\neg\neg a=a$; it does not force $\neg$ to be involutive, so 
Theorem~\ref{thm:involutive} does not apply. The status of 
$\mathrm{RmbC}(\mathrm{ce})$ is therefore still open: what we show 
below is that the filtration technique of this paper does not extend 
to it, not that no finite paraconsistent model can exist.

The general construction of Lemma~\ref{lem:general-transfer} rests on 
choosing a finite support $R_\neg$ containing $S$ such that 
$\neg_B[R_\neg]\subseteq A$. For $(\mathrm{ce})$, this is not enough: 
to verify $a\leq\neg_A\neg_A a$ at a point $a\in R_\neg$, one needs 
$\neg_A$ to be correct not only at $a$ but at $\neg_B a$ as well---and, 
inductively, at $\neg_B(\neg_B a)$, and so on. Concretely, if 
$a\in R_\neg$ but $\neg_B a\notin R_\neg$, the construction forces 
$\neg_A(\neg_B a)={\sim}(\neg_B a)$, and the required inequality 
$a\leq{\sim}(\neg_B a)$ does not follow from the axioms of 
$\mathrm{RmbC}(\mathrm{ce})$ or from any universal property of BALFIs.

The natural remedy would be to enlarge $R_\neg$ to a finite set closed 
under $\neg_B$. But the variety $\mathsf{BI}$ is not locally finite 
(Proposition~\ref{prop:non-locally-finite}), and a BALFI satisfying 
$(\mathrm{ce})$ can contain elements whose orbit under $\neg_B$ is 
infinite. In such cases, the closure of a finite set under $\neg_B$ is 
necessarily infinite, and therefore cannot serve as the basis for a 
finite Boolean subalgebra: any finite set we choose will omit some 
value of the form $\neg_B^k(a)$, and at that boundary the definition 
of $\neg_A$ as the Boolean complement breaks $(\mathrm{ce})$. We have 
also explored modifying the off-support definition of $\neg_A$ (for 
instance, setting $\neg_A a=1$ or $\neg_A a=0$ off the support), but 
in every case this either loses the required inequality or violates 
one of the basic BALFI identities.

Consequently, the algebraic filtration developed in this paper does 
not establish the finite model property for $\mathrm{RmbC}(\mathrm{ce})$. 
The obstacle is not a technicality that a more careful choice of 
parameters could resolve, but a genuine limitation of the technique: 
the axiom $(\mathrm{ce})$ imposes a non-trivial relation between each 
element and its double negation, and preserving that relation in a 
finite algebra would require closing the support under an operation 
that can generate infinitely many values. Whether 
$\mathrm{RmbC}(\mathrm{ce})$ is decidable therefore remains open.

\subsection{Decidability of the principal consistency extensions}
\begin{corollary}\label{cor:decidability-extensions}
Let $Ax$ be a nonempty subset of the fourteen principal axioms. Then $\mathrm{RmbC}(Ax)$ is decidable whenever:
\begin{enumerate}
\item[(i)] $Ax$ contains one of the six minimal explosive pairs;
\item[(ii)] $Ax\subseteq\{(\mathrm{ciw}),(\mathrm{ci}),(\mathrm{cl}),(\mathrm{cf}),(d1),(d2),(\mathrm{a}_\circ),(i1),(\mathrm{a}_\neg)\}$;
\item[(iii)] $Ax\subseteq\{(\mathrm{cf}),(\mathrm{ca}_\wedge),(\mathrm{ca}_\vee),(\mathrm{ca}_\to)\}$ with $Ax \cap\{(\mathrm{ca}_\wedge),(\mathrm{ca}_\vee),(\mathrm{ca}_\to)\}\neq\emptyset$;
\item[(iv)] $Ax\subseteq\{(i2),(\mathrm{cf})\}$ with $(i2)\in Ax$.
\end{enumerate}
\end{corollary}
\begin{proof}
Items (ii)--(iv) follow from Theorems~\ref{thm:transfer-basic}, \ref{thm:transfer-ca}, and~\ref{thm:transfer-i2}, respectively: in each case $\mathrm{RmbC}(Ax)$ has the FMP with respect to $\mathsf{BI}(Ax)$, and since $Ax$ is finite, $\mathrm{RmbC}(Ax)$ is finitely axiomatizable, so decidability follows by Theorem~\ref{thm:harrop}. Item (i) follows the same way from Remark~\ref{rem:explosive-fmp}.
\end{proof}
\begin{remark}
Combinations reducing, via the identities of Section~\ref{sec:varieties}, to a case already covered by (ii)--(iv) inherit decidability from it. Besides $\mathsf{BI}(i1,i2)=\mathsf{BI}(\mathrm{ci})$ and $\mathsf{BI}(d1,d2)=\mathsf{BI}(\mathrm{ciw})$, this also covers $(i2)$ combined with $(\mathrm{ci})$, since $(\mathrm{ci})$ already implies $(i2)$, and $(i2)$ combined with $(d2)$, since $\mathsf{BI}(i2)\subseteq\mathsf{BI}(d2)$; in both cases the extra axiom does not shrink the variety further.
\end{remark}

\begin{remark}
After excluding combinations already covered by (i)--(iv) above or reduced to them by the identities of Section~\ref{sec:varieties}, the remaining combinations are open. These include: 
\begin{itemize}
    \item $(i2)$ combined with $(\mathrm{ciw})$, $(\mathrm{cl})$, $(d1)$, $(\mathrm{a}_\circ)$, or $(\mathrm{a}_\neg)$;
    \item $(i2)$ combined with any $(\mathrm{ca}_\#)$;
    \item any combination of an axiom from Theorem~\ref{thm:transfer-basic} with $(\mathrm{ca}_\#)$ that does not already contain one of the six minimal explosive pairs.
\end{itemize}
  For $(\mathrm{ce})$, aside from combinations already covered by an explosive pair, everything beyond $\{(\mathrm{ce}),(\mathrm{cf})\}$ remains open. The pair $\{(\mathrm{ce}),(\mathrm{cf})\}$ itself is not merely open: by Corollary~\ref{cor:ce-cf}, no finite BALFI validates it paraconsistently, while $\mathrm{RmbC}(\mathrm{ce},\mathrm{cf})$ is itself paraconsistent (Example~\ref{ex:infinite-ce-cf}); the formula witnessing that fact is therefore valid on every finite model in $\mathsf{BI}(\mathrm{ce},\mathrm{cf})$ yet not a theorem, so the finite model property fails outright, for any technique. In particular, the decidability of $\mathrm{RC_1^{i2}}$ itself (Section~\ref{sec:rcila}) remains open, since it combines $(i2)$, $(\mathrm{a}_\circ)$, and $(\mathrm{a}_\neg)$ with $(\mathrm{ca}_\#)$.
\end{remark}

\begin{remark}
%The boundary above has a uniform explanation. Theorem~\ref{thm:transfer-basic} relies on the off-support default $\circ_A\equiv 1$; $(\mathrm{ca}_\#)$ requires instead the default $0$ (Theorem~\ref{thm:transfer-ca}); and $(i2)$ requires the identity default $f_\circ(a)=a$ (Theorem~\ref{thm:transfer-i2}). No single global off-support default among these three choices can satisfy all the corresponding requirements in general: $(d1)$ forces $\circ_Aa\vee a=1$ off the support, which default $0$ turns into $a=1$; $(\mathrm{ciw})$ forces $\circ_Aa={\sim}(a\wedge{\sim}a)=1$ off the support, which the identity default turns into $a=1$ --- both false for arbitrary $a$. Reconciling these requirements would require a construction sensitive to the specific element rather than a global rule, and therefore a genuinely different algebraic treatment. We leave this for future work.

%We stress that ``open'' above refers strictly to the finite model property via the present technique, not to paraconsistency: by the classification of Section~\ref{sec:paraconsistency}, every pair among the fourteen axioms other than the six explosive pairs and $\{(\mathrm{ce}),(\mathrm{cf})\}$ already admits a finite paraconsistent witness. The gap here is between paraconsistency, which is settled, and decidability via algebraic filtration, which is not.
The remaining open cases reflect a limitation of the present filtration
method. The transfer constructions require different off-support
defaults for $\circ$: the constant $1$ for the axioms covered by
Theorem~\ref{thm:transfer-basic}, the constant $0$ for
$(\mathrm{ca}_\#)$, and the identity for $(i2)$. These requirements are
in general incompatible with a single uniform completion rule, so
handling combinations of these axioms requires a different
construction.
\end{remark}

\section{Computational Complexity of $\mathrm{RmbC}$}
\label{sec:complexity}

The decidability results of the preceding sections do not by 
themselves yield a satisfactory complexity bound. In this section we establish a \textsf{coNP}-hardness lower bound and
a \textsf{coNEXPTIME} upper bound for the validity problem of
$\mathrm{RmbC}$. We then determine the exact complexity of the
principal extensions that collapse to the classical core identified
in Section~\ref{sec:classification}, and establish a uniform lower
bound for the remaining principal extensions.

\subsection{Complexity-theoretic preliminaries}
\label{sec:complexity-prelim}

For a logic $L$ over $\mathrm{For}(\Sigma)$, we write
\[
\mathrm{Th}(L):=\{\varphi\in\mathrm{For}(\Sigma) \ : \ \vdash_L\varphi\},
\qquad
\mathrm{NonTh}(L):=\mathrm{For}(\Sigma)\setminus\mathrm{Th}(L),
\]
for the sets of theorems and non-theorems of $L$, respectively. For a 
formula $\varphi$, $|\varphi|$ denotes its size, measured by the 
number of connective occurrences in $\varphi$; we write $n$ for 
$|\varphi|$ throughout this section. We use the following notation: $f(n)=O(g(n))$ means that there exist
constants $C>0$ and $n_0$ such that $f(n)\leq Cg(n)$ for all
$n\geq n_0$.

We assume familiarity with the standard complexity classes
$\textsf{NP}$, $\textsf{coNP}$, $\textsf{EXPTIME}$,
$\textsf{NEXPTIME}$, $\textsf{coNEXPTIME}$, and
$k\text{-}\textsf{EXPTIME}$, as well as with polynomial--time reducibility (see, e.g., \cite{Sipser2013}).

\subsection{Lower bound: \textsf{coNP}-hardness}
\label{sec:complexity-lower}
To establish the lower bound, we first show that $\mathrm{RmbC}$ is conservative over the positive fragment $\mathrm{CPL}^+$, and then use this conservativity to transfer the known $\textsf{coNP}$-hardness of classical validity to $\mathrm{RmbC}$.
\begin{proposition}[\textbf{Conservativity} over $\Sigma^+$]\label{prop:conservative}
For every formula $\varphi\in\mathrm{For}(\Sigma^+)$,
\[
\vdash_{\mathrm{CPL}^+}\varphi
\quad\text{if and only if}\quad
\vdash_{\mathrm{RmbC}}\varphi.
\]
\end{proposition}

\begin{proof}
The first direction follows from the fact that $\mathrm{RmbC}$ is an extension of $\mathrm{CPL}^+$, i.e., every theorem of $\mathrm{CPL}^+$ is also a theorem of $\mathrm{RmbC}$. 

For the second direction, suppose $\nvdash_{\mathrm{CPL}^+}\varphi$. By completeness of $\mathrm{CPL}^+$ with respect to the classical two-valued semantics, there exists a classical valuation $v_0:\mathrm{Var}(\varphi)\to\{0,1\}$ such that $v_0(\varphi)=0$. Let 
$\mathbf{A}=\langle\{0,1\},\wedge,\vee,\to,\neg_A,\circ_A,0,1\rangle$ where $\wedge,\vee,\to$ are the standard two-valued operations and
\[
\neg_A 0=1,\quad \neg_A 1=0,\quad \circ_A 0=1,\quad \circ_A 1=1.
\]
It is straightforward to see that $\mathbf{A}\in\mathsf{BI}$.
%: for $a=0$, $0\vee 1=1$ and $0\wedge 1\wedge 1=0$; for $a=1$, $1\vee 0=1$ and $1\wedge 0\wedge 1=0$. 
Define $v:\mathrm{For}(\Sigma)\to\{0,1\}$ by 
$v(p)=v_0(p)$ for $p\in\mathrm{Var}(\varphi)$, and $v(p)=0$ otherwise, 
extended homomorphically to $\mathrm{For}(\Sigma)$ by using the operators in $\mathbf{A}$. Since $\varphi\in\mathrm{For}(\Sigma^+)$, 
$v(\varphi)=v_0(\varphi)=0$. This shows that $\nvDash_{\mathsf{BI}}\varphi$. By soundness of $\mathrm{RmbC}$ with respect to BALFI semantics (Theorem~\ref{thm:completeness-RmbC}), 
$\nvdash_{\mathrm{RmbC}}\varphi$.
\end{proof}

\begin{corollary}[\textsf{coNP}-hardness]\label{cor:conp-hard}
$\mathrm{Th}(\mathrm{RmbC})$ is \textsf{coNP}-hard.
\end{corollary}

\begin{proof}
The identity map on $\mathrm{For}(\Sigma^+)$ is a polynomial-time
many-one reduction from $\mathrm{Th}(\mathrm{CPL}^+)$ to
$\mathrm{Th}(\mathrm{RmbC})$, by Proposition~\ref{prop:conservative}.
The validity problem for the purely implicational fragment of classical propositional logic is \textsf{coNP}-complete \cite{Stalmarck1989}. Since this fragment is
contained in $\mathrm{CPL}^+$, validity for $\mathrm{CPL}^+$ is
\textsf{coNP}-hard. Therefore, by Proposition~\ref{prop:conservative},
$\mathrm{Th}(\mathrm{RmbC})$ is \textsf{coNP}-hard.
\end{proof}

\subsection{An upper bound}
\label{sec:complexity-upper}
We first establish a quantitative version of the finite model property, which will be used to obtain the upper bound. The following technical result will be useful.

\begin{lemma} \label{lem:complex}
    Let $\varphi$ be a formula over $\Sigma$, and let \(\Phi=\mathrm{Sub}(\varphi)\). Then,  $|\Phi|\leq 2n+1$.
\end{lemma}
\begin{proof}
    By induction on the complexity of $\varphi$, taking into account that $\Sigma$ only has connectives of arity 1 or 2.
\end{proof}

\begin{lemma}[Quantitative finite model property] \label{lem:upper-bound-fmp-rmbc}
There exists a constant \(c>0\) such that, for every formula \(\varphi\)
with \(|\varphi|=n\), if \(\nvdash_{\mathrm{RmbC}}\varphi\), then there exist
a finite BALFI \(\mathbf{A}\) and a valuation \(v\) such that
\[
v(\varphi)\neq 1
\qquad\text{and}\qquad
|A|\leq 2^{2^{c\cdot n}}.
\]
\end{lemma}

\begin{proof}
Assume \(\nvdash_{\mathrm{RmbC}}\varphi\). By completeness of \(\mathrm{RmbC}\)
with respect to BALFIs (Theorem~\ref{thm:completeness-RmbC}), there exist a BALFI \(\mathbf{B}\) and a valuation
\(v_B:\mathrm{For}(\Sigma)\to B\) such that \(v_B(\varphi)\neq 1\).
Let \(\Phi=\mathrm{Sub}(\varphi)\) and \(S=\{v_B(\psi):\psi\in\Phi\}\).
Let \(G=S\cup\neg_B[S]\cup\circ_B[S]\), and let \(A\) be the Boolean
subalgebra of the Boolean reduct of \(\mathbf{B}\) generated by \(G\).
By the construction in the proof of Theorem~\ref{thm:fmp-rmbc}, the
algebra \(\mathbf{A}=\langle A,\wedge,\vee,\to,\neg_A,\circ_A,0,1\rangle\)
is a finite BALFI and there is a valuation \(v:\mathrm{For}(\Sigma)\to A\)
such that \(v(\varphi)=v_B(\varphi)\neq 1\).

It remains to bound $|A|$. If $n=0$ (i.e., $\varphi$ is a single propositional variable $p$), the two-element BALFI of the proof of Proposition~\ref{prop:conservative} refutes $\varphi$ under the valuation sending $p$ to $0$, and its size $2$ satisfies
$2\leq2^{2^{c\cdot 0}}$ for every $c>0$, so we may assume $n\geq1$ from
here on. Since $|\Phi|\leq 2n+1$, by Lemma~\ref{lem:complex}, we have
\(|S|\leq 2n+1\), and therefore
\[
|G|\leq |S|+|\neg_B[S]|+|\circ_B[S]|\leq 3(2n+1).
\]
A Boolean algebra generated by \(k\) elements has at most \(2^{2^k}\)
elements. Hence
\[
|A|\leq 2^{2^{|G|}}\leq 2^{2^{3(2n+1)}}\leq 2^{2^{c\cdot n}}
\]
for some constant \(c>0\) (e.g.\ \(c=9\) suffices for all \(n\geq 1\)).
\end{proof}
A straightforward brute-force decision procedure based on Lemma~\ref{lem:upper-bound-fmp-rmbc} yields a $3$-\textsf{EXPTIME} upper bound. Let $N(n):=2^{2^{c\cdot n}}$ be the bound on the size of a finite countermodel given by the lemma.
On a universe of size at most $N(n)$, each of the operators
$\neg$ and $\circ$ can be interpreted in $N(n)^{N(n)}$ different ways,
since each is a unary function from the universe to itself. Thus the
brute-force enumeration of the possible interpretations already has
size
\[
N(n)^{N(n)}
=
\left(2^{2^{c\cdot n}}\right)^{2^{2^{c\cdot n}}},
\]
which is triple exponential in $n$. Checking the BALFI equations and
the relevant valuations does not increase this asymptotic bound.
Hence the resulting decision procedure runs in $3$-\textsf{EXPTIME}.

The previous procedure is expensive because it enumerates the interpretations of the two non-classical operators over the entire algebra, whose size is already doubly exponential in $n$. Yet a formula of size $n$ only ever probes the algebra at the $O(n)$ points that arise as values of its subformulas: everywhere else, the two operators are free to take whichever value trivially satisfies the BALFI equations, the Boolean complement for $\neg$, the constant $1$ for $\circ$. This suggests replacing the two operators, as functions on the whole algebra, by their restriction to that small set of points, describing the rest of the algebra only through its atoms, whose number is singly exponential in $n$ rather than doubly. 

The following lemma makes this replacement precise: any assignment of the two operators on an arbitrary subset, extended by the two canonical defaults outside it, is automatically a BALFI, so it will suffice to guess and verify data on that subset alone.

\begin{lemma}[Absorption]\label{lem:absorption}
Let $\mathbf{A}=\langle A,\wedge,\vee,\to,0,1\rangle$ be a Boolean 
algebra, $S\subseteq A$, and $\neg_A,\circ_A:A\to A$ 
operators satisfying:
\begin{enumerate}
\item[(i)] $a\vee\neg_A a=1$ and $a\wedge\neg_A a\wedge\circ_A a=0$ 
for all $a\in S$;
\item[(ii)] $\neg_A a={\sim}a$ and $\circ_A a=1$ for all $a\in A
\setminus S$,
\end{enumerate}
where ${\sim}a$ denotes the Boolean complement of $a$ in $A$. Then 
$\langle A,\wedge,\vee,\to,\neg_A,\circ_A,0,1\rangle$ is a BALFI.
\end{lemma}

\begin{proof}
By hypothesis, the $\Sigma_{\mathrm{BA}}$-reduct is a Boolean algebra. For $a\in S$, both 
equations hold by (i). For $a\in A\setminus S$, by (ii): 
$a\vee\neg_A a=a\vee{\sim}a=1$, and $a\wedge\neg_A a\wedge\circ_A a=
(a\wedge{\sim}a)\wedge 1=0$.
\end{proof}
We now use Lemma~\ref{lem:absorption} to show $\mathrm{NonTh}(\mathrm{RmbC})\in\textsf{NEXPTIME}$ by exhibiting a certificate for non-validity. Intuitively, the certificate is a compact sketch of a countermodel, recording only the values of the subformulas of $\varphi$ together with the corresponding values of $\neg$ and $\circ$.
\begin{proposition}[\textsf{coNEXPTIME} upper bound]\label{prop:nexptime}
$\mathrm{NonTh}(\mathrm{RmbC})\in\textsf{NEXPTIME}$, and consequently 
$\mathrm{Th}(\mathrm{RmbC})\in\textsf{coNEXPTIME}$.
%\subseteq\textsf{2\text{-}EXPTIME}$.
\end{proposition}
\begin{proof}
Let $\varphi$ have $n$ connective occurrences and $\Phi = \mathrm{Sub}(\varphi)$, so $|\Phi|\le 2n+1$, by Lemma~\ref{lem:complex}.

\emph{Atom count.} By Lemma~\ref{lem:upper-bound-fmp-rmbc}, any countermodel can 
be taken to be a finite BALFI $A$ with $|A|\leq N(n)=2^{2^{c\cdot n}}$. Since 
$|A|=2^{|\mathrm{At}(A)|}$, the atom count satisfies $m:=
|\mathrm{At}(A)|\leq\log_2 N(n)=2^{c\cdot n}$, which is singly exponential in $n$. We 
identify $A\cong\wp([m])=\wp(\{1,2,...,m\})$, so each element of $A$ can therefore be represented by
a bitstring of length $m$.

\emph{Certificate.} A certificate for $\varphi\in\mathrm{NonTh}
(\mathrm{RmbC})$ consists of: an integer $m$ with $1\leq m\leq 2^{c\cdot n}$ 
in binary; for each $\psi\in\Phi$, a bitstring $s_\psi\in\{0,1\}^m$, representing the intended value $v(\psi)$; 
for each $\alpha\in\Phi$, bitstrings 
$\eta_\alpha,\kappa_\alpha\in\{0,1\}^m$, representing the intended
values $\neg_A v(\alpha)$ and $\circ_A v(\alpha)$, respectively, where $\{0,1\}^m$ denotes the set of bitstrings of length $m$. Since $|\Phi|=O(n)$ and $m=2^{O(n)}$, the certificate has total length
$2^{O(n)}$.

\emph{Verification.} The verifier checks that the guessed data are
consistent with a valuation in a BALFI refuting $\varphi$. Four checks
suffice:
\begin{itemize}
\item[\textbf{(C1)}] \emph{Compositional coherence.} The values $s_\psi$ satisfy the composition rules for subformulas: for each 
$\psi\in\Phi$, if $\psi=\alpha\#\beta$ then $s_\psi=s_\alpha\#s_\beta$; 
if $\psi=\neg\alpha$ then $s_\psi=\eta_\alpha$; if $\psi=\circ\alpha$ 
then $s_\psi=\kappa_\alpha$.
\item[\textbf{(C2)}] \emph{BALFI equations on the support}: for each 
$\alpha\in\Phi$, $s_\alpha\cup\eta_\alpha=[m]$ and $s_\alpha\cap
\eta_\alpha\cap\kappa_\alpha=\emptyset$.
\item[\textbf{(C3)}] \emph{Functionality.} Since different subformulas may have the same value, the corresponding
operator values must agree: for all 
$\alpha,\beta\in\Phi$ with $s_\alpha=s_\beta$: $\eta_\alpha=\eta_\beta$ 
and $\kappa_\alpha=\kappa_\beta$.
\item[\textbf{(C4)}] \emph{Refutation}: $s_\varphi\neq[m]$.
\end{itemize}
The verification runs in time $2^{O(n)}$. Indeed, C1 and C2 require $O(n\cdot m)$ bit operations, C3 can be checked by comparing all pairs of subformulas and therefore requires $O(n^2\cdot m)$ bit operations, and C4 requires only $O(m)$ bit operations. Since $m=2^{O(n)}$, all these bounds are $2^{O(n)}$; hence the total verification time is $2^{O(n)}$.

\emph{Soundness of the certificate procedure.} Suppose that a certificate satisfies (C1)--(C4). Let $A:=\wp([m])$ 
with the standard Boolean structure. Set $S:=\{s_\alpha:\alpha\in\Phi\}$ 
and define $\neg_A,\circ_A:\wp([m])\to\wp([m])$ by 
$\neg_A a:=\eta_\alpha$ if $a=s_\alpha$ for some $\alpha\in\Phi$, else 
$[m]\setminus a$; and $\circ_A a:=\kappa_\alpha$ if $a=s_\alpha$, else 
$[m]$. These definitions are well-defined by (C3). By Lemma~\ref{lem:absorption}, 
$\mathbf{A}$ is a BALFI: condition (i) holds on $S$ by (C2), while condition (ii) holds on $A\setminus S$ by
construction. 

We now define $v:\mathrm{For}(\Sigma)\to\wp([m])$ by $v(p)=s_p$ for atomic $p\in\Phi$, $v(p)=\emptyset$ otherwise, extended homomorphically. By (C1), a structural induction on $\psi\in\Phi$ gives
$v(\psi)=s_\psi$ for every $\psi\in\Phi$. By (C4), $v(\varphi)=s_\varphi\neq[m]=1_A$. Since 
$\mathbf{A}\in\mathsf{BI}$ and $v(\varphi)\neq1_A$, 
Theorem~\ref{thm:completeness-RmbC} gives $\nvdash_{\mathrm{RmbC}}\varphi$, 
i.e., $\varphi\in\mathrm{NonTh}(\mathrm{RmbC})$.

\emph{Completeness of the certificate procedure.} By 
Lemma~\ref{lem:upper-bound-fmp-rmbc}, there exist $B\in\mathsf{BI}$ and $v_B$ 
with $v_B(\varphi)\neq 1$ and $|B|\leq 2^{2^{c\cdot n}}$. Let 
$m=|\mathrm{At}(B)|$. Since $|B|=2^m$ and $|B|\leq 2^{2^{c\cdot n}}$, we have $m\leq 2^{c\cdot n}$. Identify $B\cong\wp([m])$, 
and set $s_\psi:=v_B(\psi)$, $\eta_\alpha:=\neg_B(v_B(\alpha))$, 
$\kappa_\alpha:=\circ_B(v_B(\alpha))$ for all $\psi,\alpha\in\Phi$. 
Conditions (C1)--(C4) hold since $v_B$ is a homomorphism, $B$ is a 
BALFI, $\neg_B,\circ_B$ are functions, and $v_B(\varphi)\neq 1_B=[m]$. 
Since $|\Phi|=O(n)$ and $m=2^{O(n)}$, the resulting certificate has
size $2^{O(n)}$.

\end{proof}

\begin{remark}
There remains a considerable gap between the
\textsf{coNP} lower bound and the \textsf{coNEXPTIME} upper bound. We leave both this question and
the possibility of strengthening the \textsf{coNP} lower bound as
open problems.
\end{remark}
\subsection{Complexity of the principal extensions}\label{sec:complexity-extensions}
We now turn to the computational complexity of the principal axiomatic
extensions of $\mathrm{RmbC}$. We first establish a uniform $\textsf{coNP}$ 
lower bound for every nonempty combination of the fourteen principal axioms. We then show that, for the extensions
containing one of the six minimal explosive pairs, this lower bound is
matched by a $\textsf{coNP}$ upper bound, yielding an exact complexity
classification for these extensions.

The lower bound follows by an argument analogous to Proposition~\ref{prop:conservative}. The same two-element BALFI used
there belongs to $\mathsf{BI}(Ax)$ for every nonempty subset $Ax$ of the
fourteen principal axioms, since it satisfies all fourteen principal
axioms by Remark~\ref{rem:K-all-axioms}.
\begin{proposition}[Uniform $\textsf{coNP}$ lower bound]\label{prop:uniform-conp-hard}
Let $Ax$ be any nonempty subset of the fourteen principal axioms.
Then $\mathrm{Th}(\mathrm{RmbC}(Ax))$ is
\textsf{coNP}-hard.
\end{proposition}
\begin{proof}
Let $\varphi\in\mathrm{For}(\Sigma^+)$ with $\nvdash_{\mathrm{CPL}^+}\varphi$. As in the proof of Proposition~\ref{prop:conservative}, there is a classical valuation $v_0:\mathrm{Var}(\varphi)\to\{0,1\}$ with $v_0(\varphi)=0$, and the two-element BALFI $\mathbf{A}=\langle\{0,1\},\wedge,\vee,\to,\neg_A,\circ_A,0,1\rangle$, with $\neg_A0=1,\neg_A1=0,\circ_A0=\circ_A1=1$, together with the corresponding extension $v$ of $v_0$, satisfies $v(\varphi)=0\neq1$.
Clearly, $\mathbf{A}\in K$. By Remark~\ref{rem:K-all-axioms}, $\mathbf{A}\in\mathsf{BI}(Ax)$. Hence, $\nvDash_{\mathsf{BI}(Ax)}\varphi$ and, by Theorem~\ref{thm:completeness-uniform}, $\nvdash_{\mathrm{RmbC}(Ax)}\varphi$.

Thus, the identity map on $\mathrm{For}(\Sigma^+)$ is a polynomial-time many-one reduction from $\mathrm{Th}(\mathrm{CPL}^+)$ to $\mathrm{Th}(\mathrm{RmbC}(Ax))$. As in Corollary~\ref{cor:conp-hard}, validity for the purely implicational fragment of $\mathrm{CPL}$ is \textsf{coNP}-complete~\cite{Stalmarck1989} and this fragment is contained in $\mathrm{CPL}^+$, so validity for $\mathrm{CPL}^+$ is \textsf{coNP}-hard. Therefore $\mathrm{Th}(\mathrm{RmbC}(Ax))$ is \textsf{coNP}-hard.
\end{proof}
\begin{corollary}[\textsf{coNP}-completeness of the explosive extensions]
\label{cor:explosive-complexity}
Let $Ax$ be any nonempty subset of the fourteen principal axioms
containing one of the six minimal explosive pairs. Then
$\mathrm{Th}(\mathrm{RmbC}(Ax))$ is \textsf{coNP}-complete.
\end{corollary}
\begin{proof}
By Corollary~\ref{cor:explosive-fmp},
$\mathsf{BI}(Ax)=K$. Hence, by Remark~\ref{rem:K-eCPL} and Theorem~\ref{thm:completeness-uniform},
\[
\vdash_{\mathrm{RmbC}(Ax)}\varphi
\quad\text{if and only if}\quad
\vdash_{e\mathrm{CPL}}\varphi.
\]
Since theoremhood in $e\mathrm{CPL}$ is polynomial-time reducible to
classical propositional validity, it follows that
\[
\mathrm{Th}(\mathrm{RmbC}(Ax))\in\textsf{coNP}.
\]
Together with Proposition~\ref{prop:uniform-conp-hard}, this yields
that $\mathrm{Th}(\mathrm{RmbC}(Ax))$ is \textsf{coNP}-complete.
\end{proof}

\begin{remark}
For principal extensions that remain paraconsistent, the preceding
proposition provides a uniform $\textsf{coNP}$ lower bound. The
$\textsf{coNEXPTIME}$ upper bound established in
Proposition~\ref{prop:nexptime}, however, does not automatically
transfer to these extensions. The additional axioms must also be verified in the completed algebra, and this may impose constraints on elements outside
the subformula support used by the certificate construction.
\end{remark}
%Para mi: el lema de absorcion se parece muchisimo al lema o teorema de transferencia de FMP, y el teorema de transferencia que cubre los 9 axiomas usa exactamente ese default (tomando los conjuntos R_# Y S iguales y f=1). Además ese lema de absorcion es la clave para lograr la upper bound de RmbC, será que puedo lograr o "transferir" una cota semejante para esos 9 axiomas? sería muy interesante, aunque parece tedioso, quizá tenga que pensar en 9 certificados diferentes.
%Por otro lado, si esa cuenta saliera, sería solo para esos axiomas, el resto quedaría afuera, justamente porque usan un defaul diferente
\section{Concluding Remarks}
\label{sec:conclusion}

This paper investigated the extent to which the self-extensional,
paraconsistent character of $\mathrm{RmbC}$ is preserved under
axiomatic extensions, addressing this question along three 
complementary axes: the algebraic classification of which axiomatic 
extensions preserve paraconsistency, the finite model property and 
decidability of $\mathrm{RmbC}$ and its principal extensions, and the 
computational complexity of the validity problem.

On the classification side (Section~\ref{sec:classification}), we 
analyzed the subvariety relations generated by the fourteen principal
consistency axioms and identified the minimal explosive combinations 
among them. We identified six minimal explosive combinations, all of which collapse to a single algebraic variety. This classification also identifies a natural $C_1$-like LFI with
replacement, $\mathrm{RC_1^{i2}}$, which remains paraconsistent
(Theorem~\ref{thm:rcila-minus}), whereas $\mathrm{RCila}$ is explosive.

A separate structural boundary was identified for the combination 
$\{(\mathrm{ce}),(\mathrm{cf})\}$: no finite BALFI can satisfy both 
axioms while remaining paraconsistent (Corollary~\ref{cor:ce-cf}). The obstruction is strictly finitary, since infinite paraconsistent
models of $\{(\mathrm{ce}),(\mathrm{cf})\}$ do exist
(Example~\ref{ex:infinite-ce-cf}). Thus, this combination provides a genuine obstruction to any
finite-model approach to paraconsistency.

On the decidability side (Sections~\ref{sec:fmp} 
and~\ref{sec:transfer}), we proved the finite model property for 
$\mathrm{RmbC}$ with respect to BALFI semantics via an algebraic 
filtration (Theorem~\ref{thm:fmp-rmbc}), settling a question left 
open in~\cite{CarnielliConiglioFuenmayor2022}. This result is 
non-trivial precisely because the variety $\mathsf{BI}$ is not 
locally finite (Proposition~\ref{prop:non-locally-finite}), so 
decidability could not be obtained from general local-finiteness 
arguments alone. The same framework transfers the finite model property to broad classes of
axiomatic extensions obtained from the fourteen principal consistency
axioms. In particular, our transfer results cover several families of
combinations of these axioms, while further combinations inherit
decidability through the algebraic identities established in
Section~\ref{sec:varieties}. 

The sole principal axiom not covered individually by the transfer results
is $(\mathrm{ce})$. The filtration technique developed here does not
extend to $\mathrm{RmbC}(\mathrm{ce})$, since verifying $(\mathrm{ce})$
would require closing the support under the full orbit of $\neg_B$, which
may be infinite. The decidability of $\mathrm{RmbC}(\mathrm{ce})$ therefore
remains open, as do several combinations of the principal axioms not
covered by the transfer results.

On the complexity side (Section~\ref{sec:complexity}), we established a
$\textsf{coNP}$-hardness lower bound and a
$\textsf{coNEXPTIME}$ upper bound for the validity problem of
$\mathrm{RmbC}$. The upper bound improves on the naive
$3$-\textsf{EXPTIME} bound obtained directly from the finite-model
construction, via the Absorption Lemma~\ref{lem:absorption}. For the
axiomatic extensions, we obtained a uniform $\textsf{coNP}$ lower bound
for every nonempty combination of the fourteen principal axioms, while
all extensions containing one of the six minimal explosive pairs are
$\textsf{coNP}$-complete. For the remaining paraconsistent extensions,
the $\textsf{coNEXPTIME}$ upper bound for $\mathrm{RmbC}$ does not
automatically transfer, so their precise complexity remains open.

Beyond these specific results, the algebraic filtration developed here
provides a general method for establishing the finite model property for
logics whose algebraic semantics are given by varieties of Boolean
algebras with non-classical operators, provided that the relevant axioms
can be preserved by a suitable extension of the finite filtered algebra,
as formalized in Lemma~\ref{lem:general-transfer}.

Several questions remain open. At the level of decidability, the main
ones are the status of $\mathrm{RmbC}(\mathrm{ce})$ and of the
axiomatic combinations not covered by the present transfer results,
including $\mathrm{RC_1^{i2}}$. At the complexity level, the exact
complexity of $\mathrm{RmbC}$ remains open, as does the possibility of
transferring the $\textsf{coNEXPTIME}$ upper bound to the remaining
paraconsistent extensions.

\section*{Acknowledgements}

This research was financed by Fundação de Amparo à Pesquisa do Estado de São Paulo (FAPESP, Brazil), grants 2020/16353-3 and 2024/13413-6. Coniglio also acknowledges support by an individual research grant from the National Council for Scientific and Technological Development (CNPq, Brazil), grant 309830/2023-0.

\end{document}